\documentclass[12pt]{article}
\usepackage{hslTR}
\usepackage{fullpage}
\newenvironment{proof}      
{\par\noindent\textbf{Proof}}{\eop\smallskip\vskip 3 pt}            
\newcommand{\eop}      
           {\hspace*{\fill}{$\vcenter{\hrule height1pt       
                     \hbox{\vrule width1pt height5pt       
            \kern5pt \vrule width1pt} \hrule height1pt}$} }   

\usepackage{graphicx}
\usepackage{amsmath}
\usepackage{amssymb}
\usepackage{wasysym}
\usepackage[all]{xy}
\usepackage[title]{appendix}
\usepackage{setspace}
\usepackage[dvipsnames]{xcolor}
\usepackage{cite}
\usepackage{amsfonts}
\usepackage{multirow}
\usepackage{subfigure}
\usepackage{mathrsfs}

\allowdisplaybreaks
\usepackage{tikz}
\usetikzlibrary{automata, positioning, arrows}
\tikzset{
		->, 
		>=stealth', 
		initial text=$ $, 
}
\input{rgsMacros.sty}

\definecolor{darkgreen}{rgb}{0,0.5,0}

\newtheorem{theorem}{Theorem}[section]
\newtheorem{definition}[theorem]{Definition}
\newtheorem{proposition}[theorem]{Proposition}
\newtheorem{assumption}[theorem]{Assumption}
\newtheorem{example}[theorem]{Example}
\newtheorem{remark}[theorem]{Remark}
\newtheorem{corollary}[theorem]{Corollary}
\begin{document}

\ititle{\LARGE \bf
Linear Temporal Logic for Hybrid Dynamical Systems: Characterizations and Sufficient Conditions}
\iauthor{
Hyejin Han \\ 
{\normalsize hhan@ucsc.edu}\\
Ricardo G. Sanfelice \\
{\normalsize ricardo@ucsc.edu} \\
}	
\idate{June 15, 2020} 
%
\makeititle
\thispagestyle{empty}
\pagestyle{empty}
\newpage
\tableofcontents
\newpage
\begin{abstract}
This paper introduces operators, semantics, characterizations, and solution-independent conditions to guarantee temporal logic specifications for hybrid dynamical systems.
Hybrid dynamical systems are given in terms of differential inclusions -- capturing the continuous dynamics -- and difference inclusions -- capturing the discrete dynamics or events -- with constraints.
State trajectories (or solutions) to such systems are parameterized by a hybrid notion of time.
For such broad class of solutions, the operators and semantics needed to reason about temporal logic are introduced.
Characterizations of temporal logic formulas in terms of dynamical properties of hybrid systems are presented -- in particular, forward invariance and finite time attractivity.
These characterizations are exploited to formulate sufficient conditions assuring the satisfaction of temporal logic formulas --- when possible, these conditions do not involve solution information.
Combining the results for formulas with a single operator, ways to certify more complex formulas are pointed out, in particular, via a decomposition using a finite state automaton.
Academic examples illustrate the results throughout the paper.
\end{abstract}
%
\section{Introduction}

\subsection{Background}
High-level languages are useful in formulating specifications for dynamical
systems that go beyond classical asymptotic stability, where convergence to the
desired point or set is typically certified to occur in the limit as time approaches infinity, that is,
over an infinitely long time horizon; see,
e.g.,~\cite{tabuada2006linear,kloetzer2008fully,kwon2008ltlc}.
Temporal logic 
employs operators and logic to define formulas that the solutions or executions to the systems should satisfy after some finite time, or during a particular amount of bounded time. 
In particular, temporal logic can be efficiently employed to determine safety and liveness type properties, where the former type of property typically guarantees that the state remains in a particular set, while the latter property that the state reaches a specific set in finite time. 
Such specifications are given in terms of a language that employs logical and temporal connectives (or operators) applied to propositions and their combinations. For certain classes of dynamical systems, verification of these properties can be performed using model checking tools. For instance, the question of whether a safety-type specification is satisfied can be answered by finding an execution that violates the specification in finite time.

Linear temporal logic (LTL), as introduced in~\cite{Pnueli.77},
permits to formulate specifications that involve temporal properties of
computer programs; see also~\cite{MannaPnueli.92}.
Numerous contributions pertaining to modeling, analysis, design, and
verification of LTL specifications for dynamical systems have appeared in the
literature in recent years. Without attempting to present a thorough review of
the very many articles in such topic, it should be noted that
in~\cite{Fainekos.ea.09.Automatica}, the authors employ temporal logic for
motion planning of mobile robots.
In their setting, the robots are given by continuous-time systems with second-order dynamics and the proposed temporal logic specifications model reachability, invariance, sequencing, and obstacle avoidance.
Similar approaches but for dynamical systems given in discrete time, which are
more amenable to computational tools, such as model checking, have also been
pursued in the literature. In~\cite{karaman2008optimal}, the authors propose
mixed integer linear programming and quadratic programming tools for the design
of algorithms required to satisfy LTL specifications for dynamical systems with
both continuous-valued and discrete-valued states.
Such systems are called
\textit{mixed logic dynamical systems} and are expressive enough to model discontinuous and hybrid piecewise discrete-time linear systems.
In~\cite{Wolff.ea.14.ICRA}, for discrete-time nonlinear systems with
continuous-valued and discrete-valued states, the authors formulate optimization
problems related to trajectory generation with linear temporal logic
specifications for which mixed integer linear programming tools are applied.
In~\cite{dimitrova2014deductive}, the design of controllers to satisfy
alternating-time temporal logic (ATL*), which is an expressive branching-time
logic that allows for quantification over control strategies, is pursued using
barrier and Lyapunov functions for a class of continuous-time systems.
More recently, using similar programming tools, in~\cite{Saha.Julius.16.ACC},
tools to design reactive controllers for mixed logical dynamical systems so as
to satisfy high-level specifications
given in the language of metric temporal logic are proposed,
while in~\cite{bisoffi2018hybrid} a hybrid system model and tools for the
satisfaction of 
a linear temporal logic specification for the trajectories of a physical plant
modeled as a continuous-time system are presented.
Promising extensions of these techniques to the case of specifications that
need to hold over pre-specified bounded horizons, called signal temporal logic,
have been recently pursued in several articles; see,
e.g.,~\cite{Raman.ea.15.HSCC}, to just list a few.

\subsection{Motivation}

Tools for the systematic study of temporal logic properties in dynamical
systems that have solutions (or executions) changing continuously over
intervals of ordinary continuous time and, at certain time instances, having
jumps in their continuous-value and discrete-valued states, such as the
frameworks proposed
in~\cite{vanderSchaftSchumacher00,LygerosJohanssonSimicZhangSastry03,Collins04,HaddadChellaboinaNersesov06,goebel2012hybrid},
are much less developed.
A hybrid system $\mathcal{H} = (C,F,D,G)$ exhibiting such behavior can be described as
follows~\cite{goebel2012hybrid}:
\begin{equation}
\begin{split}
  \dot{x} \phantom{\mbox{}^+} &\in F(x) \qquad\qquad x \in C\\
  x^+ &\in G(x) \qquad\qquad x \in D
\end{split}\label{eqn:H}
\end{equation}
where $x \in \mathcal{X}$ is the state and $\mathcal{X} \subset \mathbb{R}^n$ is the state space. The map $F : \mathcal{X}  \rightrightarrows \mathcal{X}$ is a set-valued map and denotes the flow map capturing the continuous dynamics on the flow set $C$, and  $G : \mathcal{X}  \rightrightarrows \mathcal{X}$ is a set-valued map and defines the jump map capturing the discrete dynamics on the jump set $D$.
Throughout the paper, we assume $C \subset \dom F$, $D \subset \dom G$, and $\overline{C} \cup D \cup G(D) \subset \mathcal{X}$, where $\overline{C}$ denotes its closure.

A canonical academic example of a hybrid system is the well-known bouncing ball system, which has
infinitely many events over a bounded ordinary time horizon (i.e.,~Zeno) at
time instances that are not pre-specified and actually depend on the initial
condition of the system; see, e.g.,~\cite{Johansson99}.
Another canonical example is the dynamical system resulting from controlling the temperature of a room with a logic controller, in which the jumps of the logic variables in the controller occur when the temperature hits certain thresholds.
In such hybrid dynamical systems, the study of temporal logic using discretization-based approaches may not be fitting as, in principle, the time at which a jump occurs is not known a priori and these times are likely to occur aperiodically.
Though results enabling the reasoning of continuously changing systems and signals using discrete-time methods 
are available in the literature (see, e.g.,~\cite{fainekos2009robustness}),
the sampling effect may prevent one from being able to guarantee that the properties certified
for the discretization extend to the actual continuous time process.


\subsection{Contributions}

In this paper, we present tools that permit guaranteeing high-level specifications for solutions to hybrid dynamical systems without requiring the computation of the solutions themselves or discretization of the dynamics, but rather, guaranteeing properties of the data defining the system and the existence of Lyapunov-like functions.

We consider a broad class of hybrid dynamical systems, in which the state vector may include physical and continuous-valued variables, logic and discrete-valued variables, timers, memory states, and others; solutions may not be unique and may not necessarily exist for arbitrary long hybrid time (namely, solutions may not be complete); and solutions may exhibit Zeno behavior.
In particular, as in~\cite{goebel2012hybrid}, a hybrid dynamical system is
defined by a flow map, which is given by a set-valued map governing the
continuous change of the state variables, a flow set, which is a subset of the
state space on which solutions are allowed to evolve continuously, a jump map,
which is also given by a set-valued map governing the discrete change of the
variables, and a jump set, which defines the set of points where jumps can
occur. Along with the state space, these four objects define the data of a hybrid dynamical system.

For this broad class of hybrid dynamical systems, 
characterizations of formulas involving one temporal operator and atomic propositions 
are presented in terms of dynamical properties of hybrid systems,
in particular, forward pre-invariance and finite time attractivity.
These notions are used to formulate sufficient conditions for the satisfaction of basic temporal logic formulas.
More precisely, we show that the specifications using the always operator can
be guaranteed to hold under mild conditions on the data of the hybrid system
when a forward invariance property of an appropriately defined set holds. To
arrive to such conditions, we present sufficient conditions for forward
(pre-)invariance of closed sets in hybrid dynamical systems that extend those
in~\cite{185}.
To derive conditions that certify that formulas using the eventual operator
hold, we generate results to certify finite time attractivity of sets in hybrid
dynamical systems, for which we exploit and extend the ideas used to certify finite time stability of hybrid dynamical systems in~\cite{188}.
Furthermore, our (mostly solution-independent) approach allows us to provide an estimate of the (hybrid) time it takes for a temporal specification to be satisfied, with the estimate only depending on a Lyapunov function and the initial condition of the solution being considered.
Moreover, we introduce sufficient conditions for certain formulas that combine more than one temporal operator, which combine our conditions for the individual temporal operators.

While many of our results do not require computing solutions to the hybrid dynamical
system, which is a key advantage when compared to methods for continuous-time,
discrete-time, and mixed logic dynamical systems cited above and the method for
hybrid traces in~\cite{cimatti2015hreltl}, the price to pay when using the
results in this paper is finding a certificate for finite time attractivity,
which is in terms of a Lyapunov function. It should be noted that though our
conditions are weaker than those in~\cite{dimitrova2014deductive}, finding such
functions might be challenging at times. However, the same complexity is
present in Lyapunov methods for certifying asymptotic stability of a point or a
set~\cite{Khalil}, or for employing continuously differentiable barrier
certificates and Lyapunov functions to certify temporal logic constraints for
continuous-time systems. On the other hand, it should be noted that the
framework for hybrid dynamical systems considered here is such that, under mild
conditions, in addition to enabling a converse theorem for asymptotic
stability, has robustness properties to small perturbations, which may permit
extending the results in this paper to the case under perturbations;
see~\cite[Chapters 6 and 7]{goebel2012hybrid}.

This paper significantly extends our previous work in~\cite{176} which contains no proofs, fewer results, and much less details.
This paper provides  characterizations of temporal operators for hybrid systems, some of which are equivalent, and some others shed light on necessity.
Moreover, sufficient conditions for temporal logic formulas that have more than one operator are presented in more detail than in~\cite{176}. In particular, we show how to derive conditions for formulas that have more than one operator by combining the conditions for formulas that have one operator. Additionally, a discussion on the decomposition of temporal logic formulas using finite state automata is included.
Furthermore, detailed proofs are included and more examples are provided.

The remainder of this paper is organized as follows. After preliminaries in Section~\ref{sec:preliminaries}, Section~\ref{sec:LTL_H} introduces LTL for hybrid systems. The characterizations of temporal operators using dynamical properties are presented in Section~\ref{sec:characterizations_operators}. The sufficient conditions to guarantee the satisfaction of LTL formulas are presented in Section~\ref{sec:sufficient_single} (for a single operator) and in Section~\ref{sec:conditions_combining} (for more than one operator).

\section{Preliminaries}
\label{sec:preliminaries}

In this paper, properties of a hybrid system $\mathcal{H}$ as in \eqref{eqn:H} are specified with LTL formulas, and conditions to guarantee the satisfaction of LTL formulas for $\mathcal{H}$ are presented.
A solution $\phi$ to $\mathcal{H}$ is parameterized by $(t, j) \in \mathbb{R}_{\geq 0} \times \mathbb{N}$, where $t$ is the ordinary time variable, $j$ is the discrete jump variable, $\mathbb{R}_{\geq 0} := [0, \infty)$, and $\mathbb{N} :=  \{0,1,2,\ldots\}$.
The domain $\dom\phi \subset \mathbb{R}_{\geq 0} \times \mathbb{N}$ of $\phi$ is a hybrid time domain if for every $(T,J) \in \dom\phi$, the set $\dom\phi \cap ([0,T] \times \{0,1,\ldots,J\})$ can be written as the union of sets $\bigcup_{j=0}^J (I_j \times \{j\})$, where $I_j := [t_j, t_{j+1}]$ for a time sequence $0 = t_0 \leq t_1 \leq t_2 \leq \cdots \leq t_{J+1}$. The $t_j$'s with $j > 0$ define the time instants when the state of the hybrid system jumps and $j$ counts the number of jumps.
A solution is given by a hybrid arc.
A function $\phi : E \rightarrow \mathbb{R}^n$ is a hybrid arc if $E$ is a hybrid time domain and if for each $j \in \mathbb{N}$, the function $t \mapsto \phi(t,j)$ is locally absolutely continuous on the interval $I^j = \{t:(t,j) \in E\}$.
A hybrid arc $\phi$ is a solution to $\mathcal{H} = (C,F,D,G)$ if $\phi(0,0)\in\overline{C} \cup D$; for all $j \in \mathbb{N}$ such that $I^j := \{t : (t,j) \in \dom \phi\}$ has nonempty interior\footnote{The interior of $I^j$ is denoted as $\mbox{int}I^j$.}, $\phi(t,j) \in C$ for all $t \in \mbox{int}I^j$ and $\dot{\phi}(t,j) \in F(\phi(t,j))$ for almost all $t \in I^j$; for all $(t,j) \in \dom\phi$ such that $(t,j+1) \in \dom \phi$, $\phi(t,j) \in D$ and $\phi(t,j+1) \in G(\phi(t,j))$.
A solution to $\mathcal{H}$ is called maximal if it cannot be further extended.

For convenience, we define the range of a solution $\phi$ to a hybrid system $\mathcal{H}$ as $\rge\phi = \{ \phi(t, j) : (t, j) \in \dom\phi \}$. We also define the set of maximal solutions to $\mathcal{H}$ from the set $K$ as $\mathcal{S}_{\mathcal{H}}(K) := \{ \phi : \phi$ is a maximal solution to $\mathcal{H}$ with $\phi(0,0) \in K \}$.
See~\cite{goebel2012hybrid} for more details about hybrid dynamical systems.

\section{Linear Temporal Logic for Hybrid Dynamical Systems}
\label{sec:LTL_H}
Linear Temporal Logic (LTL) provides a framework to specify desired properties
such as \emph{safety}, i.e.,~``something bad never happens'', and
\emph{liveness}, i.e.,~``something good eventually happens''. In this section,
for a given hybrid system $\mathcal{H}$, we define operators and specify
properties of $\mathcal{H}$ with LTL formulas~\cite{105}.
We first introduce atomic propositions.
\begin{definition}[{Atomic Proposition}]
An atomic proposition $p$ is a statement on the system state $x$ that, for each $x$, $p$ is either \texttt{True} (1 or $\top$) or \texttt{False} (0 or $\bot$).
\end{definition}
A proposition $p$ will be treated as a single-valued function of $x$, that is, it will be a function $x \mapsto p(x)$. The set of all possible atomic propositions will be denoted by $\mathcal{P}$.

Logical and temporal operators are defined as follows.
\begin{definition}[{Logic Operators}]~~~
\begin{itemize}
 \item $\lnot$ is the \emph{negation} operator
 \item $\lor$ is the \emph{disjunction} operator
 \item $\land $ is the \emph{conjunction} operator
 \item $\Rightarrow$ is the \emph{implication} operator
 \item $\Leftrightarrow$ is the \emph{equivalence} operator
\end{itemize}
\end{definition}
\begin{definition}[{Temporal Operators}]~~~
\begin{itemize}
 \item $\Circle $ is the \emph{next} operator
 \item $\Diamond$ is the \emph{eventually} operator
 \item $\Box$ is the \emph{always} operator
 \item $\mathcal{U}_s$ is the strong \emph{until} operator
 \item $\mathcal{U}_w$ is the weak \emph{until} operator
\end{itemize}
\label{def:temporalOperator}
\end{definition}

Given a hybrid system $\mathcal{H}$, the semantics of LTL are defined as follows.
For simplicity, we consider the case of no inputs and state-dependent atomic propositions.
When a proposition $p$ is \texttt{True} at $(t, j) \in \text{dom} \, \phi$, i.e.,~$p (\phi (t,j)) = 1$, it is denoted by
\begin{equation}
 \phi (t, j) \, \Vdash \, p ,
\end{equation}
whereas if $p$ is \texttt{False} at $(t, j) \in \text{dom} \, \phi$, it is written as
\begin{equation}
 \phi (t, j) \, \nVdash \, p .
\end{equation}
An LTL formula is a sentence that consists of atomic propositions and operators of LTL.
An LTL formula $f$ being satisfied by a solution $(t, j) \mapsto \phi \, (t, j)$ at some time $(t, j)$ is denoted by
\begin{equation}
 (\phi, (t, j)) \, \vDash \, f ,
\end{equation}
while $f$ not satisfied by a solution $(t, j) \mapsto \phi (t, j)$ at some time $(t, j)$ is denoted by\footnote{Note that to be compatible with the literature, instead of $\Vdash$, we use $\vDash$ for a formula.}
\begin{equation}
 (\phi, (t, j)) \, \nvDash \, f .
\end{equation}

Let $p, q \in \mathcal{P}$ be atomic propositions. The semantics of LTL are defined as follows: given a solution $\phi$ to $\mathcal{H}$ and $(t,j) \in \dom\phi$
\begin{subequations}
\begin{align}
	&(\phi, (t, j)) \vDash p \: \Leftrightarrow \: \phi (t, j) \Vdash p\\
	&(\phi, (t, j)) \vDash \lnot p \: \Leftrightarrow \: (\phi, (t, j)) \nvDash p\\
	&(\phi, (t, j)) \vDash p \lor q \: \Leftrightarrow \: (\phi, (t, j)) \vDash p \:\:\mbox{or} \:\: (\phi, (t, j)) \vDash q\\
	&(\phi, (t, j)) \vDash \ocircle p \: \Leftrightarrow \: (t,j+1) \in \dom\phi \:\:\mbox{and}\:\: (\phi, (t, j+1)) \vDash p\\
	&(\phi, (t, j)) \vDash p\, \mathcal{U}_s q \: \Leftrightarrow \: \exists (t', j') \in \dom \phi, t'+j' \geq t+j \:\: \mbox{s.t.} \:\: (\phi, (t', j')) \vDash q,\\
	&\qquad\qquad\qquad\qquad\; \mbox{and} \:\: \forall \, (t'',j'') \in \dom \phi \:\: \mbox{s.t.} \:\: t+j \leq t''+j'' < t'+j', (\phi, (t'', j'')) \vDash p \nonumber\\
	&(\phi, (t, j)) \vDash p\, \mathcal{U}_w q \: \Leftrightarrow \: (\phi, (t', j')) \vDash p\ \  \forall\: (t',j') \in \dom\phi \:\: \mbox{s.t.} \:\: t'+j' \geq t+j\\
	&\qquad\qquad\qquad\qquad\; \mbox{or}\:\: (\phi, (t, j)) \vDash p\, \mathcal{U}_s q \nonumber\\
	&(\phi, (t, j)) \vDash p \land q \: \Leftrightarrow \: (\phi, (t, j)) \vDash p \:\:\mbox{and}\:\: (\phi, (t, j)) \vDash q\\
	&(\phi, (t, j)) \vDash \Box p \:\Leftrightarrow \: (\phi, (t', j')) \vDash p \quad\forall \: t'+j' \geq t+j, \: (t', j') \in \dom \phi\\
	&(\phi, (t, j)) \vDash \Diamond p \: \Leftrightarrow \: \exists (t', j') \in \dom \phi, t'+j' \geq t+j \:\: \mbox{s.t.} \:\: (\phi, (t', j')) \vDash p.
\end{align}
\end{subequations}
The same semantics of LTL are used for formulas.
For example, with a given formula $f$, $\ocircle f$ is satisfied by $\phi$ at $(t,j) \in \dom\phi$ when $(t,j+1) \in \dom\phi$ and $(\phi, (t,j+1))$ satisfies $f$.


\section{Characterizations of Temporal Operators using Dynamical Properties}
\label{sec:characterizations_operators}

In this section, we present basic necessary and sufficient conditions for the satisfaction of LTL formulas involving one temporal operator such as \emph{always} ($\Box$), \emph{eventually} ($\Diamond$), \emph{next} ($\ocircle$), and \emph{until} ($\mathcal{U}$).
We first build a set $K$ on which the atomic proposition is satisfied. Then, the satisfaction of the formula is assured by guaranteeing particular properties of the solutions to the hybrid system relative to the set $K$.

 \subsection{Characterization of $\Box$ via forward invariance}t
\label{subsec:always}

According to the definition of the $\Box$ operator, given an atomic proposition $p$, a solution $(t, j) \mapsto \phi (t,j)$ to a hybrid system $\mathcal{H} = (C,F,D,G)$ on $\mathcal{X}$ satisfies the formula
\begin{equation}
f = \Box p
\label{eqn:always}
\end{equation}
at $(t,j)$ when we have that $\phi(t',j')$ satisfies $p$ for all $t'+j' \geq t+j$ such that $(t',j') \in \dom \phi$.
The set of points in $\mathcal{X}$ satisfying an atomic proposition $p$ is given by
\begin{equation}
K :=  \{x \in \mathcal{X} : p(x) = 1\} ,
\label{eqn:K_set}
\end{equation}
which, throughout the paper, is assumed to be nonempty.

To characterize that every solution $\phi$ to $\mathcal{H}$ satisfies $f$ in \eqref{eqn:always} at each $(t,j) \in \dom \phi$, each solution starting in $K$ needs to stay in $K$ for all time.
For this purpose, we recall the definition of forward pre-invariance and then present necessary and sufficient conditions guaranteeing $f$ in \eqref{eqn:always}.

\begin{definition}[{Forward pre-Invariance}]
Consider a hybrid system $\mathcal{H}$ on $\mathcal{X}$. A set $K \subset \mathcal{X}$ is said to be forward pre-invariant for $\mathcal{H}$ if every solution $\phi \in \mathcal{S}_{\mathcal{H}} (K)$ satisfies $\rge\phi \subset K$.
\label{def:fpi}
\end{definition}

Furthermore, we are also interested in $f$ in \eqref{eqn:always} being satisfied at some $(t,j) \in \dom\phi$ (not necessarily at $(t,j) = (0,0)$).
For this purpose, we define the following notion.

\begin{definition}[{Eventually Forward pre-Invariance}]
Consider a hybrid system $\mathcal{H}$ on $\mathcal{X}$. A set $K \subset \mathcal{X}$ is said to be eventually forward pre-invariant for $\mathcal{H}$ if\footnote{A notion that does not insist on the solutions starting from $K$ can also be formulated, but it would be a departure from a forward invariance notion since such a notion would hold for solutions that do not start from $K$.} for every solution $\phi \in \mathcal{S}_{\mathcal{H}} (K)$, there exists $(t,j) \in \dom\phi$ such that $\phi(t',j') \in K$ for all $(t',j') \in \dom\phi$ such that $t'+j' \geq t+j$.
\label{def:eventual_fpi}
\end{definition}

\begin{proposition}
Given an atomic proposition $p$, the formula $f = \Box p$ is satisfied for every maximal solution $\phi$ to a hybrid system $\mathcal{H}$ at $(t,j) = (0,0)$ with $\phi(0,0) \Vdash p$ if and only if the set $K$ in \eqref{eqn:K_set} is forward pre-invariant for $\mathcal{H}$.
\label{prop:always_fpi}
\end{proposition}
\begin{proof}
$(\Rightarrow)$ Since $\Box p$ is satisfied for all solutions $\phi$ at $(t,j) = (0,0)$ and $\phi(0,0)$ satisfies $p$, we have that every solution $\phi$ to $\mathcal{H}$ satisfies that $\phi(t,j) \in K = \{x \in \mathcal{X} : p(x) = 1\}$ for all $(t,j) \in \dom\phi$.
This implies that $K$ is forward pre-invariant via the definition of forward pre-invariance of the set $K$ in  Definition \ref{def:fpi}; namely, $\rge\phi \subset K$.

$(\Leftarrow)$ Since the set $K$ is forward pre-invariant, each solution $\phi$ that starts in $K$ stays in $K$. That is, $\phi(0,0)$ satisfies $p$ and each solution $\phi$ at $(t,j)$ in the domain of each solution satisfies $p$. This implies that $f = \Box p$ is satisfied for every solution $\phi$ to $\mathcal{H}$ at $(t,j) = (0,0)$ with $\phi(0,0) \Vdash p$.
\end{proof}

\begin{proposition}
Given an atomic proposition $p$,
the formula $f = \Box p$ is satisfied for every maximal solution $\phi$ to a hybrid system $\mathcal{H}$ at some $(t,j) \in \dom\phi$ with $\phi(0,0) \Vdash p$ if and only if the set $K$ in \eqref{eqn:K_set} is eventually forward pre-invariant for $\mathcal{H}$.
\end{proposition}
\begin{proof}
$(\Rightarrow)$
By the definition of $\Box$ and the definition of solutions to $\mathcal{H}$, since every solution $\phi$ to $\mathcal{H}$ starting from $K$ satisfies $\Box p$ at some $(t,j) \in \dom\phi$, $\phi(t',j')$ satisfies $p$ for all $(t',j') \in \dom\phi$ such that $t'+j' \geq t+j$; and thus, $\phi(t',j') \in K$ for all $(t',j')$ such that $t'+j' \geq t+j$.
This implies that $K$ is forward pre-invariant after $(t,j) \in \dom \phi$.
Then, we conclude that $K$ is eventually forward pre-invariant for $\mathcal{H}$ via the definition of eventually forward pre-invariance of the set $K$ in Definition~\ref{def:eventual_fpi}.

($\Leftarrow$) Since the set $K$ is eventually forward pre-invariant, for each solution $\phi$ that starts from $K$, there exists $(t, j) \in \dom\phi$ such that $\phi(t',j') \in K$ for all $(t',j') \in \dom\phi$ such that $t'+j' \geq t+j$.
This implies that $\phi(0,0)$ satisfies $p$ and such solution $\phi$ satisfies $p$ at each $(t',j') \in \dom\phi$ such that $t'+j' \geq t+j$.
Therefore, we conclude that $f = \Box p$ is satisfied for every solution $\phi$ to $\mathcal{H}$ at $(t,j) \in \dom\phi$ with $\phi(0,0) \Vdash p$.
\end{proof}

Note that when $K$ in \eqref{eqn:K_set} is not forward pre-invariant for $\mathcal{H}$, $\Box p$ is not satisfied for all solutions $\phi$ to $\mathcal{H}$ at every $(t,j) \in \dom \phi$ with $\phi(0,0) \Vdash p$.
The following example shows the case when $\Box p$ is not satisfied for a solution $\phi$ to $\mathcal{H}$ at every $(t,j) \in \dom \phi$ with $\phi(0,0) \Vdash p$.

\begin{example}
Let an atomic proposition $p$ given by
\begin{equation}
\begin{array}{@{}ll@{}}
  p(x) = 1 & \qquad\mbox{if } x \in [0,1]\\
  p(x) = 0 & \qquad\mbox{otherwise.}
\end{array}
\end{equation}
Consider a hybrid system $\mathcal{H} = (C,F,D,G)$ on $\mathcal{X} :=  \mathbb{R}$ given by
\begin{equation}
\begin{array}{ll}
  F(x) :=  0 &\qquad \forall x \in C:=  \big[0,\tfrac{1}{2}\big]\\
  G(x) := 
  \left\{
	\begin{array}{cl}
   		2 & \:\mbox{if}\:\: x = 1\\
   		0 & \:\mbox{if}\:\: x = 2
   	\end{array}\right.
  &\qquad \forall x \in D :=  \{1\} \cup \{2\}.
\end{array}
\end{equation}
Now, pick $\phi(0,0) = 1$ so that $\phi(0,0)$ satisfies $p$. A solution $\phi$ from $\phi(0,0)$ does not satisfy $p$ after the first jump; i.e.,~$\phi(0,1) \nVdash p$; however, $\phi(0,1)$ is still in the jump set $D$ so that it jumps to $0$, and it satisfies $p$ after the second jump; i.e.,~$\phi(0,2) \Vdash p$. Furthermore, the solution $\phi$ flows after the second jump so that $\phi(t,2)$ satisfies $p$ for every $t \geq 0$.
On the other hand, there exists another solution that starts from $1$ and stays flowing there for all furture time; hence, it satisfies $f$.
This example shows that $\Box p$ is not satisfied for all solutions $\phi$ to $\mathcal{H}$ at every $(t,j) \in \dom\phi$ when $K = \{x \in \mathcal{X} : p(x) = 1\}$ is not forward pre-invariant.
\end{example}

 \subsection{Characterization of $\Diamond$ via finite time attractivity}

\label{sec:characterization_eventual}

A solution $(t, j) \mapsto \phi (t,j)$ to a hybrid system $\mathcal{H}$ satisfies the formula
\begin{equation}
f = \Diamond p
\label{eqn:f_eventually}
\end{equation}
at $(t,j) \in$ dom $\phi$ when there exists $(t',j') \in$ dom $\phi$ such that $t'+j' \geq t+j$, and $\phi(t',j')$ satisfies $p$. The same set $K$ in \eqref{eqn:K_set} is used in this section.

To guarantee that every solution $\phi$ to $\mathcal{H}$ satisfies $f$ in \eqref{eqn:f_eventually} at each $(t,j) \in \dom \phi$, the distance of each solution to $K$ should become zero at some finite $(t,j) \in \dom \phi$ so that $\phi$ reaches $K$.
Related to this property, we recall the definition of finite time attractivity (FTA) for hybrid systems and then present necessary and sufficient conditions guaranteeing the formula $f$ in \eqref{eqn:f_eventually}.
In this definition, the amount of hybrid time required for a solution $\phi$ to converge to the set $K$ is captured by a settling-time function $\mathcal{T}$ whose argument is the solution $\phi$ and its output is a positive number determining the time to converge to $K$. More precisely, given $\phi$, $\mathcal{T}(\phi) :=  \inf \{ t+j : \phi(t,j) \in K \}$ is the time to reach $K$.
Below, given $x \in \mathbb{R}^n$ and a nonempty set $K \subset \mathbb{R}^n$, $|x|_K :=  \inf_{y \in K} |x-y| $. We use $\nearrow$ to denote the limit from below.

\begin{definition}[{Finite Time Attractivity}]
A closed set $K$ is said to be finite time attractive (FTA) for $\mathcal{H}$ with respect to $\mathcal{O} \subset \overline{C} \cup D$ if for every solution $\phi \in \mathcal{S}_{\mathcal{H}}(\mathcal{O})$, $\sup_{(t,j) \in \dom\phi} t+j \geq \mathcal{T}(\phi)$, and
\begin{equation}
 \lim_{(t,j) \in \dom\phi\,:\,t+j \nearrow \mathcal{T}(\phi)} {|\phi(t,j)| }_K = 0 .
\end{equation}
Furthermore, the set $K$ is said to be finite time attractive (FTA) for $\mathcal{H}$ if so it is with respect to $\overline{C} \cup D$.
\label{def:FTA}
\end{definition}

\begin{proposition}
Given an atomic proposition $p$, the formula $f = \Diamond p$ is satisfied for every solution $\phi$ to a hybrid system $\mathcal{H}$ at $(t,j) = (0,0)$ if and only if the closed set $K$ in \eqref{eqn:K_set} is FTA for $\mathcal{H}$.
\label{prop:eventually_FTA}
\end{proposition}
\begin{proof}
$(\Rightarrow)$ Since $\Diamond p$ is satisfied for every solution $\phi$ to a hybrid system $\mathcal{H}$ at $(t,j) = (0,0)$, there exists $(t',j') \in \dom\phi$ such that $t'+j' \geq 0$ and $\phi(t',j') \in K = \{x \in \mathcal{X}: p(x) = 1\}$.
In fact, $\phi(t',j') \in K$ implies $|\phi(t',j')|_K = 0$ and $t' + j' = \mathcal{T}(\phi)$;
that is,
\[
\lim_{(t,j) \in \dom\phi\,:\,t+j \nearrow \mathcal{T}(\phi)} {| \phi(t,j)| }_K = 0
\]
with $\mathcal{T}(\phi) = t'+j'$.
This implies that $K$ is FTA via the definition of FTA of the set $K$ in Definition~\ref{def:FTA}.

$(\Leftarrow)$ Since the closed set $K$ is FTA for $\mathcal{H}$, each solution $\phi$ to $\mathcal{H}$ satisfies
\[
\lim_{(t,j) \in \dom\phi\,:\,t+j \nearrow \mathcal{T}(\phi)} {| \phi(t,j)| }_K = 0
\]
and $\sup_{(t,j) \in \dom\phi} t+j \geq \mathcal{T}(\phi)$,
where $\mathcal{T}(\phi) = t'+j'$ for some $(t',j') \in \dom \phi$.
Indeed, by the definition of the set $K$, its closedness, and the (local) absolute continuity of $\phi$ (along with the continuity of the distance function to the set $K$), there exists $(t',j') \in \dom\phi$ such that $\phi(t',j')$ satisfies $p$.
This implies that $f = \Diamond p$ is satisfied for every solution $\phi$ to a hybrid system $\mathcal{H}$ at $(t,j) = (0,0)$.
\end{proof}

  \subsection{Characterization of $\ocircle$ via properties of the data of $\mathcal{H}$}

A solution $(t,j) \mapsto \phi(t,j)$ to a hybrid system $\mathcal{H} = (C,F,D,G)$ satisfies the formula
\begin{equation}
 f = \ocircle p
\label{eqn:f_next}
\end{equation}
when we have that $\phi(t, j+1)$ satisfies $p$ for each $(t,j) \in \dom \phi$.
Here, the same set $K$ in \eqref{eqn:K_set} is used. To guarantee that every solution $\phi$ to $\mathcal{H}$ satisfies $f$ in \eqref{eqn:f_next} at each $(t,j) \in \dom \phi$, each solution needs to jump to the set $K$ at the next hybrid time; i.e.,~$\phi(t,j+1) \in K$.

\begin{proposition}
Given an atomic proposition $p$,
let the set $K$ be as in \eqref{eqn:K_set}.
The formula $f = \ocircle p$ is satisfied for all maximal solutions $\phi$ to $\mathcal{H}$ at each $(t,j) \in \dom\phi$ if and only if all of the following properties hold simultaneously:
\begin{itemize}
  \item[a)] each nontrivial solution $\phi$ to $\mathcal{H}$ is such that $\phi(0,0) \in D$; and
 \item[b)] no flows of $\mathcal{H}$ are possible from any $x \in C$; and
 \item[c)] $G(D) \subset K \cap D$; and
 \item[d)] $\overline{C} \subset D$.
\end{itemize}
\label{prop:next}
\end{proposition}
\begin{proof}($\Rightarrow$) Suppose that $\ocircle p$ is satisfied for all solutions to $\mathcal{H}$. We need to show that a), b), c), and d) hold.
By definition of $\ocircle$ and definition of solutions to $\mathcal{H}$, since every solution $\phi$ to $\mathcal{H}$ satisfies $\ocircle p$, $\phi(0,0) \in D$ and $\phi(0,1) \in K$ for every $\phi(0,0) \in \overline{C} \cup D$.
In fact, if $\overline{C} \setminus D$ were not to be empty, then there would exist a (trivial) solution $\phi$ with $\phi(0,0) \notin D$, so $\ocircle p$ would not hold since $(0,1) \notin \dom\phi$.
Hence, $\overline{C} \subset D$ and $\phi(0,0) \in D$ hold; and thus, items a) and d) hold.
Next, we show that item b) holds.
Proceeding by contradiction, if flow is possible from a point $x \in C$, then there exists a solution $\phi$ to $\mathcal{H}$ with $\phi(0,0) = x$ and there exists $\varepsilon > 0$ such that $[0,\varepsilon) \times \{0\} \subset \dom\phi$. Since $x \in D$ due to $\overline{C} \subset D$, $\phi(0,0) \in D$. However, $(0,1) \notin \dom\phi$ since $[0,\varepsilon) \times \{0\} \subset \dom\phi$. This is a contradiction, and thus, item b) holds.
Finally, we show that item c) holds.
By definition of $\ocircle$, since every solution $\phi$ to $\mathcal{H}$ satisfies $\ocircle p$, then $(t, j+1) \in \dom\phi$ and $\phi(t, j+1) \in K$ for each $(t,j) \in \dom\phi$.
By definition of solutions, it implies that for each $(t,j) \in \dom\phi$, $\phi(t,j) = \xi \in D$ and $G(\xi) \subset K$. Hence, item c) holds.

($\Leftarrow$)
Note that $\phi(0,0) \in D$ and $(0,1) \in \dom\phi$ by items a) and b).
Then, by item c),
$G(\phi(0,0)) \subset K$ since $ \phi(0,0) \in D$.
Furthermore, for each $(t,j) \in \dom\phi$ such that $\phi(t,j) \in \overline{C} \cup D$, no flows are possible from $\phi(t,j)$ by items b) and d).
Thus, $(t,j+1) \in \dom\phi$ and $\phi(t,j+1) \in K$ by item c).
Therefore, $f = \ocircle p$ is satisfied for every solution $\phi$ to $\mathcal{H}$.
\end{proof}

  \subsection{Characterization of $\mathcal{U}$ via properties of the data of $\mathcal{H}$}

\label{sec:characterization_until}

According to the definition of the $\mathcal{U}_s$ operator, a solution $(t,j) \mapsto \phi(t,j)$ to a hybrid system $\mathcal{H} = (C,F,D,G)$ satisfies the formula
\begin{equation}
 f = p \: \mathcal{U}_s \, q
\label{eqn:until}
\end{equation}
at $(t,j) \in \dom\phi$ when there exists $(t',j') \in \dom \phi$ such that $t' + j' \geq t + j$ and $\phi(t',j')$ satisfies $q$; and $\phi(t'',j'')$ satisfies $p$ for all $(t'',j'') \in \dom \phi$ such that $t + j \leq t''+j'' < t'+j'$.
The set of points in $\mathcal{X}$ satisfying an atomic proposition $p$ or an atomic proposition $q$ are respectively given by
\begin{equation}
 P = \{ x \in \mathcal{X} : p(x) = 1\} \:\:\:\mbox{and}\:\:\:
 Q = \{ x \in \mathcal{X} : q(x) = 1\} .
\label{eqn:pq_sets}
\end{equation}

To guarantee that a solution $\phi$ to $\mathcal{H}$ satisfies $f$ in \eqref{eqn:until} at $(t,j) = (0,0)$, the solution needs to start and stay in the set $P$ at least until convergence to the set $Q$ happens; or the solution needs to start from the set $Q$.
For this purpose, we present necessary and sufficient conditions, related to the $\Diamond$ operator in Section~\ref{sec:characterization_eventual}, for convergence to the set $Q$.

\begin{proposition}
Given atomic propositions $p$ and $q$, let the sets $P$ and $Q$ be given in \eqref{eqn:pq_sets}.
Suppose the set $Q$ is nonempty.
The formula $f = p \,\mathcal{U}_s q$ is satisfied for all solutions $\phi$ to $\mathcal{H}$ at $(t,j) = (0,0)$ with $\phi(0,0) \Vdash p$ or $\phi(0,0) \Vdash q$ if and only if
every solution $\phi \in \mathcal{S}_{\mathcal{H}}(P \cup Q)$ satisfies that there exists $(t,j) \in \dom\phi$ for which all of the following properties hold simultaneously:
\begin{itemize}
  \item[a)] $\phi(t,j) \in Q$; and
  \item[b)] $\phi(t',j') \in P$ for all $(t',j') \in \dom\phi$ such that $t'+j' < t+j$.
\end{itemize}
\label{prop:until_strong}
\end{proposition}
\begin{proof}({$\Rightarrow$})
Suppose $p \,\mathcal{U}_s q$ is satisfied for all solutions $\phi$ to $\mathcal{H}$ at $(t,j) = (0,0)$ with $\phi(0,0) \Vdash p$ or $\phi(0,0) \Vdash q$.
Then, we need to show that every solution $\phi \in \mathcal{S}_{\mathcal{H}}(x)$ with $x \in P \cup Q$ satisfies that there exists $(t,j) \in \dom\phi$ satisfying items a) and b).
Since $\phi(0,0)$ satisfies $p$ or $q$, $\phi(0,0) \in P \cup Q$. By the definition of $\mathcal{U}_s$ and the definition of solutions to $\mathcal{H}$, since $p \,\mathcal{U}_s q$ is satisfied for every solution $\phi$ to $\mathcal{H}$ at $(t,j) = (0,0)$, 
$\phi(0,0) \in \overline{C} \cup D$ and
there exists $(t,j) \in \dom\phi$ such that
\begin{itemize}
 \item[1)] $\phi(t,j)$ satisfies $q$; and
 \item[2)] $\phi(t',j')$ satisfies $p$ for all $(t',j') \in \dom \phi$ such that $t'+j' < t+j$.
\end{itemize}
Since there exists $(t,j) \in \dom\phi$ such that $\phi(t,j) \in Q$ by item 1), item a) holds.
Additionally, by item 2), for all $(t',j') \in \dom\phi$ such that $t'+j' < t+j$, $\phi(t',j') \in P$; and thus, item b) holds.

($\Leftarrow$)
The proof is straightforward using the definition of the $\mathcal{U}_s$ operator.
\end{proof}

  \section{Sufficient Conditions for Temporal Formulas with One Operator using Hybrid Systems Tools}

\label{sec:sufficient_single}

\subsection{Sufficient Conditions for $\Box p$}

\label{subsec:sufficient_always}

In this section, we present sufficient conditions guaranteeing $f$ in \eqref{eqn:always}.
Due to the equivalence we provide in Section~\ref{subsec:always}, any sufficient condition that guarantees the needed invariance property of the set guarantees the satisfaction of the formula $\Box p$.
For example, in~\cite{185,Maghenem.Sanfelice.18.}, such invariance property for
hybrid systems is studied as follows:
\begin{itemize}
 \item Forward pre-invariance of a set in~\cite[Theorem 4.3]{185};
 \item Forward pre-invariance of a subset of the sublevel sets of a Lyapunov
  function in~\cite[Theorem 5.1]{185};
 \item Forward pre-invariance of a set defined by a barrier function
  in~\cite[Theorem 1]{Maghenem.Sanfelice.18.}.
\end{itemize}
By exploiting the results and the ideas in~\cite{Maghenem.Sanfelice.18.}, the
conditions given below provide sufficient conditions to verify that
$\mathcal{H}$ is such that every solution $\phi$ to $\mathcal{H}$ with
$\phi(0,0) \Vdash p$ satisfies $f = \Box p$.
Below, the concept of tangent cone of a set is used; see~\cite[Definition
5.12]{goebel2012hybrid}.
The tangent cone at a point $x \in \mathcal{X}$ of a set $K \subset \mathcal{X}$, denoted $T_K(x)$, is the set of all vectors $w \in \mathcal{X}$ for which there exists $x_i \in K$, $\tau_i > 0$ with $x_i \rightarrow x$ and $\tau_i \searrow 0$ such that $w = \tfrac{x_i-x}{\tau_i}$.
For a set $K \subset \mathcal{X}$, $U(K)$ denotes any open neighborhood of $K$ and $\partial K$ denotes its boundary.
Furthermore, the notion of barrier function candidate with respect to
$K$ for $\mathcal{H}$ is given as
follows~\cite{Maghenem.Sanfelice.18.}:
\begin{definition}[{Barrier Function Candidate}]
Consider a hybrid system $\mathcal{H} = (C,F,D,G)$ on $\mathcal{X}$. A function $B : \mathcal{X} \rightarrow \mathbb{R}$ is said to be a barrier function candidate with respect to $K$ for $\mathcal{H}$ if
\begin{equation}
\left\{
\begin{array}{ll}
 B(x) \leq 0 &\qquad\forall x \in K\\
 B(x) > 0 &\qquad\forall x \in (C \cup D \cup G(D)) \setminus K \mbox{.}
\end{array}
\right.
\label{eqn:barrier_candidate}
\end{equation}
\end{definition}

\begin{assumption}
The flow map $F$ is outer semicontinuous, nonempty, and locally bounded with convex images on $C$. Furthermore, the jump map $G$ is nonempty on $D$.
\label{assump:always}
\end{assumption}
\begin{theorem}
Consider a hybrid system $\mathcal{H} = (C,F,D,G)$ on $\mathcal{X}$ satisfying Assumption~\ref{assump:always}.
Given an atomic proposition $p$, suppose the state space $\mathcal{X}$ and the atomic proposition $p$ are such that $K$ in \eqref{eqn:K_set} is closed and $K \subset C \cup D$.
Then, the formula $f = \Box p$ is satisfied for all solutions $\phi$ to $\mathcal{H}$ (and for all $(t,j) \in \dom\phi$) with $\phi(0,0) \Vdash p$ if there exists a barrier function candidate $B$ with respect to $K$ for $\mathcal{H}$ as in \eqref{eqn:barrier_candidate} that is continuously differentiable and the following properties hold:
\begin{itemize}
 \item[1)] $\left\langle  \nabla B(x), \eta \right\rangle  \leq 0$ for all $x \in C \cap (U(\partial K) \setminus K)$ and all $\eta \in F(x) \cap T_C(x)$.
 \item[2)] $B(\eta) \leq 0$ for all $x \in D \cap K$ and all $\eta \in G(x)$.
 \item[3)] $G(D \cap K) \subset C \cup D$.
\end{itemize}
\label{thm:always}
\end{theorem}
\begin{proof}
Under conditions 1)--3), we conclude that the set $K$ in
 \eqref{eqn:K_set} is forward pre-invariant for $\mathcal{H}$
 using~\cite[Theorem 1]{Maghenem.Sanfelice.18.}.
Then, by Proposition~\ref{prop:always_fpi}, the formula $f = \Box p$ is satisfied for each solution $\phi$ to $\mathcal{H}$ at $(t,j) = (0,0)$ with $\phi(0,0) \Vdash p$ since the set $K$ is forward pre-invariant for $\mathcal{H}$.
Moreover, this property at $(t,j) = (0,0)$ implies $\phi(t,j) \Vdash p$ at each $(t,j) \in \dom\phi$; and thus, the formula $f = \Box p$ is satisfied for each solution $\phi$ to $\mathcal{H}$ and at each $(t,j) \in \dom \phi$ with $\phi(0,0) \Vdash p$.
\end{proof}

\begin{remark}
Note that $\Box p$ is satisfied for all solutions $\phi$ to $\mathcal{H}$ if $\phi(0,0) \Vdash p$ and $\phi(t,j) \Vdash p$ for all future hybrid time $(t,j) \in \dom \phi$.
Under the conditions in Theorem~\ref{thm:always}, solutions with $\phi(0,0) \nvDash p$ may satisfy $p$ after some time if $\phi$ reaches the set $K$ in \eqref{eqn:K_set} in finite time. Convergence to such set in finite hybrid time is presented in the next section.
\end{remark}

Next, the bouncing ball example in~\cite[Example 1.1]{goebel2012hybrid}
illustrates Theorem~\ref{thm:always}.
\begin{example}
Consider a hybrid system $\mathcal{H} = (C,F,D,G)$ modeling a ball bouncing vertically on the ground, with state $x = (x_1,x_2) \in \mathcal{X}:= \mathbb{R}^2$ and the data given by
\begin{equation}
\begin{split}
 F(x) &:= 
 \left[
\begin{array}{@{}c@{}}
  x_2\\
  -\gamma
\end{array}
 \right] \:\:\:\: \forall x \in C :=  \{x \in \mathcal{X} : x_1 \geq 0\} ,\\
 G(x) &:= 
 \left[
\begin{array}{@{}c@{}}
  0\\
  -\lambda x_2
\end{array}
 \right] \:\: \forall x \in D :=  \{x \in \mathcal{X} : x_1 = 0, x_2 \leq 0\},
\end{split}
\label{eqn:H_ball}
\end{equation}
where $x_1$ denotes the height above the surface and $x_2$ is the vertical velocity. The parameter $\gamma > 0$ is the gravity coefficient and $\lambda \in [0,1]$ is the restitution coefficient.
Every maximal solution to this system is Zeno.
Define an atomic proposition $p$ as follows: for every $x \in \mathcal{X}$, $p(x) = 1$ when $x \in C \cup D$ and $2 \gamma x_1 + (x_2 - 1) (x_2 + 1) \leq 0$; $p(x) = 0$ otherwise.
Let $K$ be given as in \eqref{eqn:K_set}. Then, we observe that the closed set $K$ is the sublevel set where the total energy of the ball is less than or equal to $1/2$.
The function $B(x) :=  2 \gamma x_1 + (x_2 - 1) (x_2 + 1)$ is a barrier function candidate since $B(x) \leq 0$ for all $x \in K$ and $B(x) > 0$ otherwise.
Then, we have $\left\langle  \nabla B(x), F(x) \right\rangle  = 0$ for each $x \in C$; and thus, condition 1) in  Theorem~\ref{thm:always} is satisfied.
Moreover, we have $B(G(x)) = 2 \gamma x_1 + \lambda^2 x_2^2 - 1 \leq 0$ for every $x \in D \cap K$ since $\lambda \in [0,1]$; and thus, condition 2) in  Theorem~\ref{thm:always} is satisfied.
Finally, since $G(D) = \{0\} \times \mathbb{R}_{\geq 0} \subset C \cup D$, condition 3) in Theorem~\ref{thm:always} is satisfied.
Therefore, via Theroem~\ref{thm:always}, the formula $f = \Box p$ is satisfied for each solution $\phi$ to $\mathcal{H}$ from $K$ and at each $(t,j) \in \dom\phi$. \hfill $\triangle$
\label{ex:bouncing_ball}
\end{example}

\begin{example}
Consider a hybrid system $\mathcal{H} = (C,F,D,G)$ modeling a constantly evolving timer system with the state $x = (\tau,h) \in \mathcal{X} :=  [0, \infty) \times \{0,1\} $ given by
\begin{equation}
\begin{split}
 F(x) &:= 
 \left[
\begin{array}{@{}c@{}}
  1\\
  0
\end{array}
 \right] \qquad\quad \forall x \in C :=  \{ x \in \mathcal{X} : 0 \leq \tau \leq T\},\\
 G(x) &:= 
 \left[
\begin{array}{@{}c@{}}
  0\\
  1-h
\end{array}
 \right] \quad\: \forall x \in D :=  \{ x \in \mathcal{X} : \tau \geq T\},
\end{split}
\end{equation}
where $\tau$ denotes a timer variable, $h$ is a logic variable, and $T$ is the period of the timer. 
Moreover, for each $x \in \mathcal{X}$ such that $0 \leq \tau \leq T$, $p(x) = 1$; otherwise, $p(x) = 0$.
Let $K$ be given as in \eqref{eqn:K_set}.
Consider the barrier function candidate $B(x) :=  \tau - T$.
We notice that $C \cap (U(\partial K) \setminus K) = \emptyset$; and thus, condition 1) in Theorem~\ref{thm:always} is trivially satisfied.
Moreover, we have $B(G(x)) = -T \leq 0$ for every $x \in D$; and thus, condition 2) in Theorem~\ref{thm:always} is satisfied.
Furthermore, since $G(D) = \{0\} \times \{0,1\} \subset C \cup D$, condition 3) in Theorem~\ref{thm:always} is satisfied.
Therefore, via Theorem~\ref{thm:always}, the formula $f = \Box p$ is satisfied for each solution $\phi$ to $\mathcal{H}$ and at each $(t,j) \in \dom \phi$. \hfill $\triangle$
\end{example}

 \subsection{Sufficient Conditions for $\Diamond p$}

\label{subsec:sufficient_eventually}

We now present sufficient conditions guaranteeing the formula $f$ in \eqref{eqn:f_eventually}.
Due to the equivalence we provide in Section~\ref{sec:characterization_eventual}, any sufficient condition that guarantees the FTA property of the set $K$ in \eqref{eqn:K_set} guarantees the satisfaction of the formula $\Diamond p$.
In that sense, we observe that the results on finite-time stability (FTS) for a set for hybrid systems in~\cite{188} and the results on recurrence for a set for hybrid systems in~\cite{subbaraman2016equivalence} can be applied to
derive sufficient conditions guaranteeing the desired FTA property.
In the following, by exploiting the results and the ideas in~\cite{188},
sufficient conditions are proposed to verify that $\mathcal{H}$ is such that
every solution $\phi$ to $\mathcal{H}$ satisfies $f = \Diamond p$;
see Appendix~\ref{appendix:eventually} for more details about sufficient conditions for FTA.

As stated above, the satisfaction of the formula $f = \Diamond p$ is assured by conditions that guarantee that the set $K$ in \eqref{eqn:K_set} is FTA for $\mathcal{H}$, where
\begin{equation}
 p(x) =
 \left\{
\begin{array}{ll}
  1 & \qquad\mbox{if}\:\: x \in K \\
  0 & \qquad\mbox{otherwise.}\end{array}\right.
\end{equation}

In the following, we propose sufficient conditions to satisfy the formula $f = \Diamond p$.
Below, the function $V : \mathcal{X} \rightarrow \mathbb{R}$ is continuous on $\mathcal{X}$ and locally Lipschitz on a neighborhood of $C$.
Using Clarke generalized derivative, we define the functions $u_C$ and $u_D$ as follows: $u_C(x) :=  \max_{{v \in F(x)}} \: \max_{\zeta \in \partial V(x)} \langle \zeta,v \rangle$ for each $x \in C$, and $-\infty$ otherwise; $u_D(x) :=  \max_{\zeta \in G(x)} V(\zeta) - V(x)$ for each $x \in D$, and $-\infty$ otherwise, where $\partial V$ is the generalized gradient of $V$ in the sense of Clarke;
see Appendix~\ref{appendix:nonsmooth_ly} for more details.
Moreover, a function $\alpha : \mathbb{R}_{\geq 0} \mapsto \mathbb{R}_{\geq 0}$ is a class-$\mathcal{K}$ function, denoted by $\alpha \in \mathcal{K}$, if it is zero at zero, continuous, and strictly increasing and $\alpha$ is a class-$\mathcal{K}_\infty$ function, denoted by $\alpha \in \mathcal{K}_\infty$, if $\alpha \in \mathcal{K}$ and is unbounded.
Given a real number $s \in \mathbb{R}$, $\mbox{ceil}(s)$ denotes the smallest integer upper bound for $s$.

\begin{theorem}
Consider a hybrid system $\mathcal{H}=(C,F,D,G)$ on $\mathcal{X}$.
Given an atomic proposition $p$, suppose the state space $\mathcal{X}$ and the atomic proposition $p$ are such that $K$ in \eqref{eqn:K_set} is closed.
Suppose there exists an open set\footnote{A set $\mathcal{N}$ can be chosen as $\mathcal{N} = \mathcal{X}$ for the global version of FTA.} $\mathcal{N}$ that defines an open neighborhood of $K$ such that $G(\mathcal{N}) \subset \mathcal{N} \subset \mathcal{X}$.
Then, if
\begin{itemize}
  \item[1)]there exists a continuous function $V:\mathcal{N} \rightarrow \mathbb{R}_{\geq 0}$, locally Lipschitz on an open neighborhood of $C \cap \mathcal{N}$, and $c_1 > 0, c_2 \in [0,1)$ such that
\begin{itemize}
   \item[1.1)] for every $x \in \mathcal{N} \cap (\overline{C} \cup D)$ such that $p(x) = 0$, each $\phi \in \mathcal{S}_{\mathcal{H}} (x)$ satisfies
\begin{equation}
    \tfrac{V^{1-c_2} (x)}{c_1 (1-c_2)} \leq \sup_{(t,j) \in \text{dom}\,\phi} t \mbox{;}
   \label{eqn:sup_t}
\end{equation}
   \item[1.2)] the function $V$ is positive definite with respect to $K$ and
\begin{itemize}
    \item[1.2a)] for each $x \in \mathcal{X}$ such that $x \in C \cap \mathcal{N}$ and $p(x) = 0$, $u_C(x) + c_1 V^{c_2} (x) \leq 0$;
    \item[1.2b)] for each $x \in \mathcal{X}$ such that $x \in D \cap \mathcal{N}$ and $p(x) = 0$, $u_D(x) \leq 0$.
\end{itemize}
\end{itemize}
\end{itemize}
or
\begin{itemize}
 \item[2)] there exists a continuous function $V:\mathcal{N} \rightarrow \mathbb{R}_{\geq 0}$, locally Lipschitz on an open neighborhood of $C \cap \mathcal{N}$, and $c > 0$ such that
\begin{itemize}
   \item[2.1)] for every $x \in \mathcal{N} \cap (\overline{C} \cup D)$ such that $p(x) = 0$, each $\phi \in \mathcal{S}_{\mathcal{H}} (x)$ satisfies
\begin{equation}
     \mbox{ceil} \left( \tfrac{V(x)}{c} \right) \leq \sup_{(t,j) \in \text{dom}\,\phi} j \mbox{;}
    \label{eqn:sup_j}
\end{equation}
   \item[2.2)] the function $V$ is positive definite with respect to $K$ and
\begin{itemize}
    \item[2.2a)] for each $x \in \mathcal{X}$ such that $x \in C \cap \mathcal{N}$ and $p(x) = 0$, $u_C(x) \leq 0$;
    \item[2.2b)] for each $x \in \mathcal{X}$ such that $x \in D \cap \mathcal{N}$ and $p(x) = 0$, $u_D(x) \leq -\min\{c, V(x)\}$.
\end{itemize}
\end{itemize}
\end{itemize}
hold, then, the formula $f = \Diamond p$ is satisfied for every solution $\phi$ to $\mathcal{H}$ from $L_V(r) \cap (\overline{C} \cup D)$ at $(t,j) = (0,0)$ where $L_V(r) = \{x \in \mathcal{X} : V(x) \leq r\}$, $r \in [0, \infty]$, is a compact sublevel set of $V$ contained in $\mathcal{N}$.
Moreover, for each $\phi \in \mathcal{S}_{\mathcal{H}} (L_V(r) \cap (\overline{C} \cup D))$, defining $\xi = \phi(0,0)$, the first time $(t',j') \in \dom\phi$ such that $\phi(t',j') \Vdash p$ satisfies
\begin{equation}
 t'+j' = \mathcal{T}(\phi) ,
\label{eqn:settling_time}
\end{equation}
and an upper bound on that hybrid time is given as follows:
\begin{itemize}
 \item[a)] if 1) holds, then $\mathcal{T}$ is upper bounded by $\mathcal{T}^\star(\xi) + \mathcal{J}^\star(\phi)$, where $\mathcal{T}^\star(\xi) = \tfrac{V^{1-c_2} (\xi)}{c_1 (1-c_2)}$ and $\mathcal{J}^\star(\phi)$ is such that $(\mathcal{T}^\star(\xi), \mathcal{J}^\star(\phi)) \in \dom\phi$.
 \item[b)] if 2) holds, then $\mathcal{T}$ is upper bounded by $\mathcal{T}^\star(\phi) + \mathcal{J}^\star(\xi)$, where $\mathcal{J}^\star(\xi) = \mbox{ceil} \bigl( \tfrac{V(\xi)}{c} \bigr)$ and $\mathcal{T}^\star(\phi)$ is such that $(\mathcal{T}^\star(\phi), \mathcal{J}^\star(\xi)) \in \dom\phi$ and $(\mathcal{T}^\star(\phi), \mathcal{J}^\star(\xi) -1) \in \dom\phi$.
\end{itemize}
\label{thm:FTA}
\end{theorem}

\begin{proof}
Note that the set $K$ is closed and collects the set of points such that $p$ is satisfied.
Now we show that the conditions in Proposition~\ref{prop:FTA_flow} or Proposition~\ref{prop:FTA_jump} hold for $K$.
\begin{itemize}
 \item Item 1) implies that for every $x \in \mathcal{N} \cap (\overline{C} \cup D) \setminus K$, each $\phi \in \mathcal{S}_{\mathcal{H}}(x)$ satisfies \eqref{eqn:sup_t};
and the function $V$ is positive definite with respect to
  $K$; and $u_C(x) + c_1 V^{c_2} (x) \leq 0$ for every $x \in (C \cap \mathcal{N}) \setminus K$ and
  $u_D(x) \leq 0$ for all $x \in (D \cap \mathcal{N}) \setminus K$.
Thus, Proposition~\ref{prop:FTA_flow} applies.
 \item Item 2) implies that for every $x \in \mathcal{N} \cap (\overline{C} \cup D) \setminus K$, each $\phi \in \mathcal{S}_{\mathcal{H}}(x)$ satisfies \eqref{eqn:sup_j};
and the function $V$ is positive definite with respect to $K$;
and $u_C(x) \leq 0$ for every $x \in (C \cap \mathcal{N}) \setminus K$ and $u_D(x) \leq -\min\{c, V(x)\}$ for every
  $x \in (D \cap \mathcal{N}) \setminus K$.
Thus, Proposition~\ref{prop:FTA_jump} applies.
\end{itemize}
Therefore, $K$ is FTA for $\mathcal{H}$ if item 1) or 2) holds.
Then, by Proposition~\ref{prop:eventually_FTA}, the formula $f = \Diamond p$ is satisfied for all solutions to $\mathcal{H}$
at $(t,j) = (0,0)$.
\end{proof}

 Note that the conditions about the supremum over the hybrid time of a solution in \eqref{eqn:sup_t} and \eqref{eqn:sup_j} are due to not insisting on completeness of maximal solutions. When every maximal solution is complete, these conditions hold automatically.
See Remark~\ref{remark:complete_solutions_FTA} for more details.

\begin{remark}
Under condition 1.2) or 2.2) in Theorem~\ref{thm:FTA}, given a solution $\phi$ to $\mathcal{H}$, there exists some time $(t',j') \in \dom\phi$ such that $\phi$ satisfies $p$.
Furthermore, we have this satisfaction in finite time $(t', j')$, obtained by the settling-time function $\mathcal{T}$, for which an upper bound depends on the Lyapunov function and the solution only.
Note that a settling-time function $\mathcal{T}$ does not need to be computed. However, we provide an estimate of when convergence happens using an upper bound that depends on $V$ and the constants involved in items 1) and 2) only.
\end{remark}

\begin{remark}
Note that conditions \eqref{eqn:sup_t} and \eqref{eqn:sup_j} hold for free for complete solutions unbounded in $t$ or/and $j$ in their domain.
Moreover, maximal solutions are complete when the conditions
 in~\cite[Proposition~2.10 or Proposition~6.10]{goebel2012hybrid} hold.
Specifically, if maximal solutions $\phi$ are complete with $\dom\phi$ unbounded in its $t$ component, then \eqref{eqn:sup_t} holds automatically; and, if the solutions are complete with $\dom\phi$ unbounded in its $j$ component, then \eqref{eqn:sup_j} holds automatically.
\label{remark:complete_solutions_FTA}
\end{remark}

\begin{remark}
Item 1) in Theorem~\ref{thm:FTA} characterizes the situation when the formula $f = \Diamond p$ is being satisfied for all solutions $\phi$ to $\mathcal{H}$ due to the strict decrease of a Lyapunov function during flows.
Item 2) in Theorem~\ref{thm:FTA} provides conditions for $f$ to be satisfied for all solutions $\phi$ to $\mathcal{H}$ due to a Lyapunov function strictly decreasing at jumps.
Finally, we can combine the properties in item 1) and item 2) to arrive to
 strict Lyapunov conditions for verifying that $\mathcal{H}$ is such that
 every $\phi$ satisfies $f$ at $(t,j) = (0,0)$;
see Proposition~\ref{prop:FTA_jump2}.
\end{remark}

\begin{remark}
Based on the definition of recurrence for sets in~\cite[Definition
 1]{subbaraman2016equivalence}, the recurrence property could be used
 for certifying the formula $\Diamond p$. When the set $K$ that
 collects the set of points such that $p$ is satisfied is
 globally recurrent for a given hybrid system $\mathcal{H} = (C,F,D,G)$, for each
 complete solution $\phi \in \mathcal{S}_{\mathcal{H}}(C \cup D)$, there exists $(t,j) \in \dom \phi$ such that
 $\phi(t,j) \in K$; namely, it implies that $\phi$ satisfies
 $p$ at $(t,j) \in \dom\phi$.
In~\cite{subbaraman2016equivalence}, robustness of  recurrence  and equivalence
 between the uniform and non-uniform notions are established for open
 and bounded sets.
We observe that the recurrence property is studied with respect to open sets.
 Therefore, once we have an open, bounded set that collects the set of
 points satisfying $p$, we can employ the recurrence property
 to verify that $\Diamond p$ is satisfied. Furthermore, we can use the
 results on robustness of recurrence presented
 in~\cite{subbaraman2016equivalence} to derive the satisfaction of the
 formula $\Diamond p$ with robustness.
\end{remark}

In the following examples, the item 1) in Theorem~\ref{thm:FTA} is exercised.

\begin{example}
Inspired from~\cite[Example 3.3]{188}, consider a hybrid system $\mathcal{H} = (C,F,D,G)$
 with state $x = (z, \tau) \in \mathbb{R} \times [0,1]$ given by\footnote{The function $\mbox{sgn} : \mathbb{R} \rightarrow \{-1,1\}$ is
 defined as $\mbox{sgn}(x) = 1$ if $x \geq 0$, and $\mbox{sgn}(x) = -1$ otherwise.}
\begin{equation}
\begin{split}
 F(x) &:= 
 \left[
\begin{array}{@{}c@{}}
  -k | z| ^\alpha \mbox{sgn}(z)\\
  1
\end{array}
 \right] \quad \forall x \in C :=  \mathbb{R} \times [0,1],\\
 G(x) &:= 
 \left[
\begin{array}{@{}c@{}}
  -z\\
  0
\end{array}
 \right] \qquad\qquad\quad\:\:\, \forall x \in D :=  \mathbb{R} \times \{1\},
\end{split}
\label{eqn:H_FTA}
\end{equation}
where $\alpha \in (0,1)$ and $k > 0$.
Consider the function $V : \mathbb{R} \times [0,1] \rightarrow \mathbb{R}_{\geq 0}$ given by $V(x) = \tfrac{1}{2} z^2$ for each $x \in C$. Moreover, each $x \in C$ satisfies $p$ only when $x \in \{0\} \times [0,1]$.
Now we consider the set $K=\{x \in C : p(x) = 1\}$.
We have that, for each $x \in C \setminus K$,
\[
 \langle \nabla V(x), F(x) \rangle = -k| z| ^{1+\alpha} = -2^{\tfrac{1+\alpha}{2}} k V(x)^{\tfrac{1+\alpha}{2}} .
\]
Furthermore, for all $x \in D \setminus K$, $V(G(x)) - V(x) = 0$. Therefore, condition 1.2) in Theorem~\ref{thm:FTA} is satisfied with $\mathcal{N} = \mathbb{R} \times [0,1]$, $c_1 = 2^{\tfrac{1+\alpha}{2}} k > 0$ and $c_2 = \tfrac{1+\alpha}{2} \in (0,1)$.
By applying~\cite[Proposition 6.10]{goebel2012hybrid}, condition 1.1) in Theorem~\ref{thm:FTA} holds since every maximal solution to $\mathcal{H}$ is
 complete with its domain of definition unbounded in the $t$
 direction. Thus, the formula $f = \Diamond p$ is satisfied for all
 solutions $\phi$ to $\mathcal{H}$ at $(t,j) = (0,0)$. \hfill $\triangle$
\label{eqn:ex_2}
\end{example}

 Next, the bouncing ball example in Example~\ref{ex:bouncing_ball} illustrates Lyapunov conditions for verifying that $\Diamond p$ is satisfied for all solutions to $\mathcal{H}$ at $(t,j)=(0,0)$.

\begin{example}
Consider $\mathcal{H} = (C,F,D,G)$ in Example~\ref{ex:bouncing_ball}.
Define an atomic proposition $p$ as follows:
for each $x \in \mathcal{X}$, $p(x) = 1$ when $x_2 \leq 0$, and $p(x) = 0$ otherwise.
With $K$ in \eqref{eqn:K_set} and $\mathcal{N} = \mathcal{X}$, let $V(x) = | x_2| $ for all $x \in \mathcal{X}$.
This function is continuously differentiable on the open set $\mathcal{X} \setminus (\mathbb{R} \times \{0\})$ and it is Lipschitz on $\mathcal{X}$.
It follows that
\[
 \langle \nabla V(x), F(x) \rangle = -\gamma \qquad\quad \forall x \in (C \cap \mathcal{N}) \setminus K,
\]
and $u_C(x) + c_1 V^{c_2} (x) \leq 0$ holds with $c_1 = \gamma$ and $c_2=0$.
For each $x \in (D \cap \mathcal{N}) \setminus K$,
\[
V(G(x))  = -\lambda x_2 = | -\lambda x_2| = \lambda | x_2|  = \lambda V(x),
\]
and $u_D(x) = V(G(x)) - V(x) = \lambda V(x) - V(x) = -(1-\lambda) V(x)$.
 Thus, condition 1.2) in Theorem~\ref{thm:FTA} is satisfied since $(D \cap \mathcal{N}) \setminus K = \emptyset$.
 Note that by applying~\cite[Proposition 6.10]{goebel2012hybrid}, every maximal
 solution is complete and condition 1.1) in Theorem~\ref{thm:FTA} holds
 with the chosen constants $c_1$ and $c_2$ due to the
 properties of the hybrid time domain of each maximal solution.
Therefore, the formula $f = \Diamond p$ is satisfied for all maximal solutions to $\mathcal{H}$ at $(t,j)=(0,0)$.
Since every solution from $K$, after some time, jumps from $K$ and then converges to $K$ again in finite time, we have that $f = \Diamond p$ holds for every $(t,j)$ in the domain of each solution. \hfill $\triangle$
\end{example}

Note that Theorem~\ref{thm:FTA} guarantees that $\Diamond p$ is satisfied for all solutions $\phi$ to $\mathcal{H}$ at $(t,j) = (0,0)$.
These conditions can be extended to guarantee that $\Diamond p$ is satisfied for all $(t,j)$ in the domain of any solution if the set $K$ is forward pre-invariant or when only jumps are allowed from points in $K$ and the jump map maps points in $K$ into $\mathcal{N}$.

\begin{theorem}
Consider a hybrid system $\mathcal{H}=(C,F,D,G)$ on $\mathcal{X}$.
Given an atomic proposition $p$, suppose the state space $\mathcal{X}$ and the atomic proposition $p$ are such that $K$ in \eqref{eqn:K_set} is closed and that there exists an open set $\mathcal{N}$ that defines an open neighborhood of $K$ such that $G(\mathcal{N}) \subset \mathcal{N} \subset \mathcal{X}$.
Then, if there exists a continuous function $V:\mathcal{N} \rightarrow \mathbb{R}_{\geq 0}$, locally Lipschitz on an open neighborhood of $C \cap \mathcal{N}$, and $c, c_1 > 0$, $c_2 \in [0,1)$,
such that each $\phi \in \mathcal{S}_{\mathcal{H}} (L_V(r) \cap (C \cup D))$ is complete,
$G(D \cap K) \subset L_V(r) \cap (C \cup D)$, and
at least one among items 1.2) and 2.2) in Theorem~\ref{thm:FTA} holds, then, the formula $f = \Diamond p$ is satisfied for every solution $\phi$ to $\mathcal{H}$ from $L_V(r) \cap (C \cup D)$ and for all $(t,j)$ in the domain of each solution, where $L_V(r) = \{x \in \mathcal{X} : V(x) \leq r\}$, $r \in [0, \infty]$ is a compact sublevel set of $V$ contained in $\mathcal{N}$.
\label{thm:FTA_all}
\end{theorem}
\begin{proof}
The set $K$ is closed and collects the points such that $p$ is satisfied.
We first show the case when item 1.2) in Theorem~\ref{thm:FTA} holds. Since each solution $\phi \in \mathcal{S}_{\mathcal{H}} (L_V(r) \cap (C \cup D))$ is complete, this implies that
there exists $(t_1, j_1) \in \dom \phi$ such that
\[
 \lim_{t+j \nearrow t_1 + j_1} | \phi(t,j)| _K = 0.
\]
If there exists $(t_2,j_2) \in \dom \phi$ such that $\phi(t_2,j_2) \notin K$, then $\phi$ left $K$ by jumping since condition 1.2a) in Theorem~\ref{thm:FTA} does not allow flowing out of $K$. However, if that is the case, then $\phi(t_2,j_2) \in L_V(r) \cap (C \cup D)$ since $G(D \cap K) \subset L_V(r) \cap (C \cup D)$; and then, due to completeness of $\phi$, there exists $(t_3, j_3)$ such that $\lim_{t+j \nearrow t_3 + j_3} | \phi(t,j)| _K = 0$.
Thus, proceeding in this way for all hybrid time instant that the solution leaves $K$,
condition 1) in Theorem~\ref{thm:FTA} holds and for every $(t,j)$ in the domain of each solution $\phi$.
Therefore, the formula $f = \Diamond p$ is satisfied for all solutions to $\mathcal{H}$ from $L_V(r) \cap (C \cup D)$ for every $(t,j)$ in the domain of each solution.
The proof for the cases when item 2.2) in Theorem~\ref{thm:FTA} holds follows similarly.
\end{proof}

The following example about the firefly model in~\cite[Example
25]{goebel2009hybrid} illustrates Theorem~\ref{thm:FTA_all}.

\begin{example}
Consider the hybrid system $\mathcal{H} = (C,F,D,G)$ modeling two impulsive oscillators capturing the dynamics of two fireflies. This system has the state $x = (x_1,x_2) \in \mathbb{R}^2$ and the data given by
\begin{equation}
\begin{split}
 F(x) &:= 
 \left[
\begin{array}{@{}c@{}}
  \gamma\\
  \gamma
\end{array}
 \right] \;\;\qquad\qquad\quad \forall x \in C :=  [0,1] \times [0,1],\\
 G(x) &:= 
 \left[
\begin{array}{@{}c@{}}
  g((1+\tilde{\varepsilon})x_1)\\
  g((1+\tilde{\varepsilon})x_2)
\end{array}
 \right] \:\:\: \forall x \in D :=  \{x \in C : \max{\{x_1,x_2\}} = 1\},
\end{split} 
\end{equation}
where $\gamma > 0$ and the parameter $\tilde{\varepsilon} > 0$ denotes the effect on the timer of a firefly when the timer of the other firefly expires, and the set-valued map $g$ is given by $g(z) = z$ when $z < 1$; $g(z) = 0$ when $z > 1$; $g(z) = \{0,1\}$ when $z = 1$.
Define $p$ as follows: for each $x \in \mathbb{R}^2$, $p(x) = 1$ when $x \in C$ and $x_1 = x_2$, and $p(x) = 0$ otherwise.
Then, the set $K$ is $\{x \in C : p(x) = 1\}$.
Let $k = \tfrac{\tilde{\varepsilon}}{2 + \tilde{\varepsilon}}$ and
note that $\tfrac{1+\tilde{\varepsilon}}{2 + \tilde{\varepsilon}} = \tfrac{1+k}{2}$.
Define
\[
V(x) :=  \min{\{ | x_1-x_2| , 1+k-| x_1-x_2| \}}
\]
for all $x \in \mathcal{X} :=  \{x \in \mathbb{R}^2 : V(x) < \tfrac{1+k}{2}\} = \{x \in \mathbb{R}^2 : | x_1 - x_2|  \neq \tfrac{1+k}{2} \}$.
This function is continuously differentiable on the open set $\mathcal{X} \setminus K$ and it is Lipschitz on $\mathcal{X}$.
Let $m^\star = \tfrac{1+k}{2}$ and $m \in (0, m^\star)$.
Consider $C_m = C \cap M$ and $D_m = D \cap M$, where $M :=  \{ z \in C \cup D : V(x) \leq m\}$.
By the definition of $V$, it follows that
\[
 \langle \nabla V(x), F(x) \rangle = 0 \qquad \forall x \in C_m \setminus K .
\]
We now consider $x \in D_m \setminus K$.
Since $V$ is symmetric, without loss of generality, consider $x = (1, x_2) \in D_m \setminus K$ where $x_2 \in [0,1] \setminus \{\tfrac{1}{2+\tilde{\varepsilon}}\}$.\footnote{Since $\left(1, \tfrac{1}{2 + \tilde{\varepsilon}}\right) \in \left\{x \in \mathcal{X} : V(x) = \tfrac{1+k}{2}\right\}$.}
Then, we obtain
\[
\begin{split}
 V(x) &= \min \{1-x_2, k+x_2\},\\
 V(G(x)) &= \min \{ g((1+\tilde{\varepsilon})x_2), 1+k-g((1+\tilde{\varepsilon}) x_2) \} .
\end{split}
\]
When $g((1 + \tilde{\varepsilon})x_2) = 0$, it follows that $V(G(x)) = 0$.
When $g((1 + \tilde{\varepsilon})x_2) = (1 + \tilde{\varepsilon}) x_2$, there are two cases:
\begin{itemize}
 \item[a)] $x_2 < \tfrac{1}{2 +\tilde{\varepsilon}}$, $V(x) = k + x_2 > (1 + \tilde{\varepsilon}) x_2 \geq V(G(x))$;
 \item[b)] $x_2 > \tfrac{1}{2 + \tilde{\varepsilon}}$, $V(x) = 1 - x_2 \geq V(G(x))$.
\end{itemize}
Thus, $V(G(x)) - V(x) \leq 0$ for all $x \in D_m \setminus K$.
By applying~\cite[Proposition 6.10]{goebel2009hybrid}, every maximal solution
 to the hybrid system $\mathcal{H}_m = (C_m, F, D_m, G)$ is complete.
Moreover, given $\tilde{\varepsilon} > 0$, for $\varepsilon = \tfrac{\tilde{\varepsilon}}{1 + \tilde{\varepsilon}}$ and $m$ such that $(K + \varepsilon \mathbb{B}) \cap C \subset C_m$, we have that for all $x \in D_m \cap (K + \varepsilon \mathbb{B})$, $G(x) = 0 \in K$.
Therefore, it follows from Theorem~\ref{thm:FTA_all} that the formula $f = \Diamond p $ is satisfied for every solution $\phi$ to $\mathcal{H}$ from $\mathcal{N} :=  \{x \in C \cup D : V(x) < m\}$ for all $(t,j)$ in the domain of each solution.
\hfill $\triangle$
\end{example}
 \subsection{Sufficient Conditions for $\ocircle p$}

\begin{theorem}
Given an atomic proposition $p$, let the set $K$ be as in \eqref{eqn:K_set}.
The formula $f = \ocircle p$ is satisfied for all solutions $\phi$ to $\mathcal{H}$ at each $(t,j) \in \dom\phi$ if the properties a), b), and c) in Proposition~\ref{prop:next} hold simultaneously.
\label{thm:next}
\end{theorem}

\begin{remark}
By the definition of \emph{next} operator, one could consider that the flow set $C$ is empty to specify $\ocircle p$ for all solutions $\phi$ to $\mathcal{H}$. Under this assumption, $\mathcal{H}$ reduces to a discrete-time system.
\end{remark}

The following example illustrates the sufficient conditions in  Theorem~\ref{thm:next} to guarantee the satisfaction of $\ocircle p$.
\begin{example}
 Let a hybrid system $\mathcal{H} = (C,F,D,G)$ with the state $x \in \mathbb{R}$ and data given by
\begin{equation}
  D :=  \mathbb{R} , \quad G(x) :=  \mbox{sgn} (x) ,
\end{equation}
$C$ is empty, and the flow map $F$ is arbitrary.
The function sgn$(x)$ is defined in Example~\ref{eqn:ex_2}, and $p(x) = 1$ if $| x|  = 1$. Let $K :=  \{-1,1\}$. By using the map $G$, for every $x \in D \cap K$, $G(x) \in K$; for every $x \in D \setminus K$, $G(x) \in K$. Therefore, the formula $f=\ocircle p$ is satisfied for all solutions to $\mathcal{H}$.
\end{example}
 \subsection{Sufficient Conditions for $p \,\mathcal{U} q$}

We present sufficient conditions guaranteeing $f$ in \eqref{eqn:until} by applying the results in Section~\ref{subsec:sufficient_always} and Section~\ref{subsec:sufficient_eventually}.
Note that the \emph{until} operator is characterized as \emph{strong until} ($\mathcal{U}_s$) or \emph{weak until} ($\mathcal{U}_w$).
First, we present sufficient conditions for the formula having the weak until operator. The following result is immediate.

\begin{theorem}
Consider a hybrid system $\mathcal{H}=(C,F,D,G)$ on $\mathcal{X}$ and two atomic propositions $p$ and $q$.
Suppose every $x \in \mathcal{X}$ satisfies either $p(x) = 1$ or $q(x) = 1$.
Then, the formula $f = p \,\mathcal{U}_w q$ is satisfied for every solution $\phi$ to $\mathcal{H}$ at every $(t,j) \in \dom\phi$.
\label{thm:until_1}
\end{theorem}
\begin{proof}
Since every $x \in \mathcal{X}$ satisfies either $p(x) = 1$ or $q(x) = 1$,
$\mathcal{X} = P \cup Q$ where $P$ and $Q$ are the sets in \eqref{eqn:pq_sets}.
It implies that for all $x \in \mathcal{X}$, $x \in P$ if $x$ does not belong to $Q$; and thus, every solution that has not converged to $Q$ remains in $P$ at least until it converges to $Q$ (if that ever happens).
Therefore, the formula $f = p \,\mathcal{U}_w q$ is satisfied for all solutions to $\mathcal{H}$ at $(t,j) \in \dom\phi$.
\end{proof}

\begin{remark}
By the definition of $\mathcal{U}_w$ operator, $p \,\mathcal{U}_w q$ is instantly satisfied by a solution $\phi$ to $\mathcal{H}$ at $(t,j) \in \dom\phi$ if $\phi(t,j) \in Q :=  \{x \in \mathcal{X}: q(x) = 1\}$. Moreover, $p \,\mathcal{U}_w q$ is satisfied by $\phi$ at $(t,j) \in \dom\phi$ if $\phi(t',j') \in P :=  \{x \in \mathcal{X}: p(x) = 1\}$ for all $(t',j') \in \dom\phi$ such that $t'+j' \geq t+j$. Under the conditions in Theorem~\ref{thm:until_1}, each solution $\phi \in \mathcal{S}_{\mathcal{H}}(x)$ belongs to either $P$ or $Q$ for each $(t,j) \in \dom\phi$. Therefore, we observe that $q \,\mathcal{U}_w p$ is satisfied for each solution $\phi$ to $\mathcal{H}$ at each $(t,j) \in \dom\phi$.
\end{remark}

Furthermore, if the conditions for FTA in Theorem~\ref{thm:FTA} with $p$ therein replaced by $q$ hold and there exists an open set $\mathcal{N}$ defining an open neighborhood of $Q$ in \eqref{eqn:pq_sets} such that $G(\mathcal{N}) \subset \mathcal{N} \subset \mathcal{X}$, then, under the assumptions in  Theorem~\ref{thm:FTA}, solutions to $\mathcal{H}$ from $L_V(r)$ are guaranteed to satisfy $q$ in finite time where $L_V(r) = \{x \in \mathcal{X} : V(x) \leq r\}$, $r \in [0, \infty]$ is a sublevel set of $V$ contained in $\mathcal{N}$.

 The following result relaxes the covering of $\mathcal{X}$ in  Theorem~\ref{thm:until_1} by requiring that $P$ contains a subset of the basin for finite-time attractivity of $Q$. It provides conditions for the formula $f = p \,\mathcal{U}_s q$ to be satisfied for all solutions $\phi$ to $\mathcal{H}$, both at $(t,j) = (0,0)$ and any $(t,j) \in \dom\phi$.

\begin{theorem}
Consider a hybrid system $\mathcal{H}=(C,F,D,G)$ on $\mathcal{X}$.
Given atomic propositions $p$ and $q$, and sets $P$ and $Q$ in \eqref{eqn:pq_sets},
suppose there exists an open set $\mathcal{N}$ defining an open neighborhood of $Q$ such that $G(\mathcal{N}) \subset \mathcal{N} \subset \mathcal{X}$.
Then, the formula $f = p \, \mathcal{U}_s q$ is satisfied for every solution $\phi$ to $\mathcal{H}$ at $(t,j)=(0,0)$ if
\begin{itemize}
 \item[1)] $Q$ is closed;
 \item[2)] at least one among condition 1) and 2) in Theorem~\ref{thm:FTA} with $p$ therein replaced by $q$ is satisfied with some function $V$ as required therein;
 \item[3)] $\phi(0,0) \in (P \cap L_V(r)) \cup Q$;
 \item[4)] $(L_V(r) \cap (\overline{C} \cup D)) \setminus Q \subset P$,
\end{itemize}
where $L_V(r)$ is a compact sublevel set of $V$ coming from item 2) that is contained in $\mathcal{N}$.
Moreover, the upper bound of the settling-time function $\mathcal{T}$ is given in item a) or b) in Theorem~\ref{thm:FTA}, respectively.
 Furthermore, if the following holds:
\begin{itemize}
 \item[5)] For each $x \in Q \cap D$, $G(x) \subset L_V(r) \cap (\overline{C} \cup D)$ where $L_V(r)$ as above,
\end{itemize}
then the formula $f = p \,\mathcal{U}_s q$ is satisfied for every solution $\phi$ to $\mathcal{H}$ at every $(t,j) \in \dom\phi$.
\label{thm:until_2}
\end{theorem}
\begin{proof}
When item 1) and 2) hold, $Q$ is FTA for $\mathcal{H}$.
Note that $L_V(r)$ is a compact sublevel set of $V$ contained in $\mathcal{N}$.
If the initial state is away from $Q$, it is in $P \cap L_V(r)$; and thus, it converges to $Q$ in finite time.
Moreover, condition 4) implies that the solutions from $L_V(r) \cap (\overline{C} \cup D)$ flow/jump to $P$ if the solutions have not converged to $Q$, which is guaranteed to occur in finite hybrid time.
At this point, the formula $f = p \,\mathcal{U}_s q$ is satisfied for all solutions to $\mathcal{H}$ at $(t,j)=(0,0)$.
Furthermore, condition 5) implies that for each $x \in Q \cap D$, jumps map the state to $L_V(r) \cap (\overline{C} \cup D)$; and thus, every solution from $x$ converges to $Q$ in finite time if it has not converged to $Q$, and the solutions that has not converged to $Q$ remains in $P$.
That is, once a solution belongs to $Q$, it stays in $Q$ or flows/jumps to $P$.
Therefore, the formula $f = p \,\mathcal{U}_s q$ is satisfied for all solutions to $\mathcal{H}$ at every $(t,j) \in \dom \phi$.
\end{proof}

Though at times might be more restrictive, condition 4) in Theorem~\ref{thm:until_2} can be replaced by forward invariance of $P$ when $C$ and $F$ satisfy condition 1) in Theorem~\ref{thm:always}.

The bouncing ball example in Example~\ref{ex:bouncing_ball} is used to illustrate Theorem~\ref{thm:until_2}.

\begin{example}
Consider $\mathcal{H} = (C,F,D,G)$ in Example~\ref{ex:bouncing_ball}.
Define $p$ as 1 when $x_2 \geq 0$, and 0 otherwise.
Define $q$ as 1 when $x_2 \leq 0$, and 0 otherwise.
With the sets $P$ and $Q$ in \eqref{eqn:pq_sets}, as shown in Example~\ref{ex:bouncing_ball}, item 2) in Theorem~\ref{thm:until_2} is satisfied with $\mathcal{N} = \mathbb{R}^2$.
By applying item 1) in Theorem~\ref{thm:FTA}, every solution from $Q$, after some time, jumps from $Q$ to $P$ and then converges to $Q$ again in finite time.
Moreover, from the definition of $Q$ and $P$ in \eqref{eqn:pq_sets}, if a solution does not belong to $Q$, then it belongs to $P$.
Furthermore, $P$ satisfies item 4) in Theorem~\ref{thm:until_2} since $(\overline{C} \cup D) \setminus Q \subset P$.
For each $x \in Q \cap D$, $\xi \in G(x)$ satisfies $\xi \in \overline{C} \cup D$ so that item 5) in Theorem~\ref{thm:until_2} is satisfied.
Thus, every solution that has not converged to $Q$ remains in $P$ at least until it converges to $Q$, which is guaranteed to occur in finite hybrid time.
We conclude that the formula $f = p \,\mathcal{U}_s q$ is satisfied for all solutions to $\mathcal{H}$ at every $(t,j) \in \dom \phi$. \hfill $\triangle$
\end{example}

\section{Sufficient Conditions for LTL Formulas Combining Operators}

\label{sec:conditions_combining}

\begin{table}[!thb]
\centering
\begin{tabular}{ll}
	\hline
	& \qquad\qquad Sufficient Conditions\\
	\hline
	$\Box p$  & \qquad\qquad a) Barrier function for forward pre-invariance\\
	$\Diamond p$ & \qquad\qquad b) Lyapunov function for FTA\\
	$p \,\mathcal{U} q$ & \qquad\qquad c) Combination of a) and b)\\
	$\ocircle p$ & \qquad\qquad d) $G(D) \subset D \cap \{ x \in \mathcal{X}: p(x) = 1\}$\\
	\hline
\end{tabular}
\caption{Sufficient conditions for $\Box$, $\Diamond$, $\mathcal{U}$, $\ocircle$}
\label{table:conditions}
\end{table}

Section~\ref{sec:sufficient_single} provides sufficient conditions for formulas that involve a single temporal operator. Table~\ref{table:conditions} summarizes the conditions for each temporal operator.
As indicated therein, all that is needed is either a certificate for finite-time convergence in terms of a Lyapunov function, or the data of the hybrid system and the set of points where the proposition is true to satisfy conditions for invariance.
The latter can be actually certified using Lyapunov-like functions or barrier
functions as in~\cite{185}, which for space reasons is not pursued here.

Moreover, the case of logic operators can be treated similarly by using intersections, unions, and complements of the sets where the propositions hold. For instance, sufficient conditions for $\Box (p \land q)$ can immediately be derived from the sufficient conditions already given in Section~\ref{subsec:sufficient_always} with $\{x \in \mathcal{X} : p(x) = 1\} \cap \{x \in \mathcal{X} : q(x) = 1\}$ in place of $\{x \in \mathcal{X} : p(x) = 1\}$.

The following sections present sufficient conditions for formulas that combine more than one operator. The conditions therein are given by compositions of the conditions in Table~\ref{table:conditions}.

\subsection{Conditions for $\Diamond \Box$}

\begin{corollary}
Consider a hybrid system $\mathcal{H}=(C,F,D,G)$ on $\mathcal{X}$ and an atomic proposition $p$.
Suppose $C$ is closed relative to $\mathcal{X}$, and
\begin{itemize}
 \item The state space $\mathcal{X}$ and the atomic proposition $p$ are such that $\{x \in \mathcal{X} : p(x) = 1\}$ is closed;
 \item The map $F : \mathcal{X} \rightrightarrows \mathcal{X}$ is outer semicontinuous, locally bounded relative to $\{x \in C : p(x) = 1\}$, and $F(x)$ is convex for every $x \in \{x \in C : p(x) = 1\}$. The map $F$ is locally Lipschitz on $\{x \in C : p(x) = 1\}$; and
 \item There exists an open set $\mathcal{N}$ that defines an open neighborhood of $\{x \in \mathcal{X} : p(x) = 1\}$ such that $G(\mathcal{N}) \subset \mathcal{N} \subset \mathcal{X}$.
\end{itemize}
Then, the formula $f = \Diamond \Box p$ is satisfied for all solutions $\phi$ to $\mathcal{H}$ for all $(t,j) \in \dom \phi$ if the following properties hold:
\begin{itemize}
 \item[1)] Conditions 1), 2), and 3) in Theorem~\ref{thm:always} hold; and
 \item[2)] Condition 1) or condition 2) in Theorem~\ref{thm:FTA} holds.
\end{itemize}
\label{cor:eventually_always_1}
\end{corollary}

Alternatively, sufficient conditions to guarantee the formula $\Diamond \Box p$ can be obtained by strengthening the Lyapunov conditions in Theorem~\ref{thm:FTA}.
\begin{corollary}
Consider a hybrid system $\mathcal{H}=(C,F,D,G)$ on $\mathcal{X}$ and an atomic proposition $p$.
Suppose
\begin{itemize}
 \item The state space $\mathcal{X}$ and the atomic proposition $p$ are such that $\{x \in \mathcal{X} : p(x) = 1\}$ is closed; and
 \item There exists an open set $\mathcal{N}$ that defines an open neighborhood of $\{x \in \mathcal{X} : p(x) = 1\}$ such that $G(\mathcal{N}) \subset \mathcal{N} \subset \mathcal{X}$.
\end{itemize}
Then, the formula $f = \Diamond \Box p$ is satisfied for all solutions $\phi$ to $\mathcal{H}$ that remain in a compact subset of $\mathcal{N}$ for all $(t,j) \in \dom \phi$ if the following properties hold:
\begin{itemize}
 \item[1)] there exists a continuous function $V:\mathcal{N} \rightarrow \mathbb{R}_{\geq 0}$, locally Lipschitz on an open neighborhood of $C \cap \mathcal{N}$, and $c, c_1 > 0$, $c_2 \in [0,1)$ such that
\begin{itemize}
    \item[1.1)] for every $x \in \mathcal{N} \cap (\overline{C} \cup D)$ such that $p(x) = 0$, each $\phi \in \mathcal{S}_{\mathcal{H}} (x)$ satisfies $\tfrac{V^{1-c_2} (x)}{c_1 (1-c_2)} \leq \sup_{(t,j) \in \dom\phi} t$ and $\mbox{ceil} \left( \tfrac{V(x)}{c} \right) \leq \sup_{(t,j) \in \dom\phi} j$;
    \item[1.2)] the function $V$ is positive definite with respect to $K$ and
\begin{itemize}
     \item[1.2a)] for each $x \in \mathcal{X}$ such that $x \in C \cap \mathcal{N}$, $u_C(x) + c_1 V^{c_2} (x) \leq 0$;
     \item[1.2b)] for each $x \in \mathcal{X}$ such that $x \in D \cap \mathcal{N}$, $u_D(x) \leq -\min\{c, V(x)\}$.
\end{itemize}
\end{itemize}
\end{itemize}
\label{cor:eventually_always_2}
\end{corollary}

Corollary~\ref{cor:eventually_always_2} imposes bounds on 1.2a) and 1.2b) for each point where flow and jump is possible, respectively, rather than only when $p$ is not true. Such conditions further guarantee invariance of $\{x \in \mathcal{X} : p(x) = 1\}$.

A similar estimate for the time to converge as in Theorem~\ref{thm:FTA} holds.
Condition 1) in Corollary~\ref{cor:eventually_always_1} can be alternatively
guaranteed with a Lyapunov-like/barrier function as in~\cite{185}.

Corollary~\ref{cor:eventually_always_2} requires strict Lyapunov functions, but nonstrict versions as in Theorem~\ref{thm:FTA} can be similarly stated.

\subsection{Conditions for $\Box \Diamond$}

Sufficient conditions to guarantee the formula $f = \Box \Diamond p$ are given by those in Theorem~\ref{thm:FTA_all}.

\subsection{Conditions for $\Box (p \,\mathcal{U}_s q)$}

Sufficient conditions to guarantee the formula $f = \Box (p \,\mathcal{U}_s q)$ are already given by those in Theorem~\ref{thm:until_2}.

\subsection{Conditions for $p \,\mathcal{U}_s \Box q$}

The formula $f = p \,\mathcal{U}_s \Box q$ can be certified by applying Theorem~\ref{thm:until_2} and Corollary~\ref{cor:eventually_always_2} with $p$ therein replaced by $q$.
\begin{corollary}
Consider a hybrid system $\mathcal{H}=(C,F,D,G)$ on $\mathcal{X}$. Suppose $C$ is closed in $\mathcal{X}$ and
\begin{itemize}
 \item There exists an open set $\mathcal{N}$ defining an open neighborhood of $\{x \in \mathcal{X} : q(x) = 1\}$ such that $G(\mathcal{N}) \subset \mathcal{N} \subset \mathcal{X}$;
 \item The map $F : \mathcal{X} \rightrightarrows \mathcal{X}$ is outer semicontinuous, locally bounded relative to $\{x \in C : p(x)=1\}$, and $F(x)$ is convex for every $\{x \in C : p(x)=1\}$. Additionally, the map $F$ is locally Lipschitz on $\{x \in C : p(x)=1\}$.
\end{itemize}
Then, the formula $f = p \, \mathcal{U}_s \Box q$ is satisfied for every solution $\phi$ to $\mathcal{H}$ if
\begin{itemize}
 \item[1)] all of the conditions in Theorem~\ref{thm:until_2} hold; and
 \item[2)] condition 1.2) in Corollary~\ref{cor:eventually_always_2} with $p$ therein replaced by $q$ holds.
\end{itemize}
\label{cor:until_always}
\end{corollary}

 \subsection{Decomposition of general formulas using finite state automata}

In certain cases, formulas that combine more than one operator can be decomposed into simpler formulas for which our results for formulas with a single operator can be applied. To decompose a general formula combining into several formulas with a single operator, one can employ the finite state automaton (FSA) representation of an LTL formula \cite{wolper2000constructing,babiak2012ltl,belta2017formal}.
Following~\cite[Chapter 2]{belta2017formal}, a particular fragment of LTL, called syntactically co-safe LTL (scLTL), is considered so that each formula $f$ over a set of observations can always be translated into an FSA.
An LTL formula belongs to the scLTL fragment if it contains only temporal operators $\Diamond$, $\ocircle$, $\mathcal{U}$, and it is written in positive normal form: the negation operator $\lnot$ occurs only in front of atomic propositions.
Next, given an LTL formula $f$ in the scLTL fragment, we outline the process of constructing an FSA, which we denote $A_f$, and specify properties of a hybrid system $\mathcal{H}$ with $A_f$. 
We first introduce the FSA representation of LTL formulas that belongs to the scLTL fragment.

\begin{definition}[{Finite State Automaton}]
Given an scLTL formula $f$, a finite state automaton (FSA) is given by the tuple $A_f = (S, s_0, O, \delta, S_F)$, where
\begin{itemize}
 \item $S$ is a finite set of states,
 \item $s_0 \in S$ is the initial state,
 \item $O$ is a finite set of observations,
 \item $\delta : S \times O \rightarrow S$ is a transition function,\footnote{When $\delta$ is set valued, namely, $\delta : S \times O \rightrightarrows S$ maps points in $S \times O$ to subsets of $S$, then $A_f$ is said to be non-deterministic.}
 \item $S_F \subseteq S$ is the set of accepting (final) states.
\end{itemize}
\end{definition}
The semantics of an FSA are defined over finite words of observations (or inputs).
A run of $A_f$ over a word of observations $w_O = w_O(1)w_O(2) \ldots w_O(n)$ with $w_O(k) \in O$ for all $k = 1, \ldots, n$ is a sequence $w_S = w_S(1)w_S(2) \ldots w_S(n+1) \in S$ where $w_S(1) = s_0$ and $w_S(k+1) = \delta(w_S(k),w_O(k))$ for all $k = 1, \ldots, n$.
The word $w_O$ is accepted by $A_f$ if the corresponding run ends in an accepting automaton state; i.e.,~$w_S(n+1) \in S_F$.

With an FSA associated to a general formula $f$ in the scLTL fragment,
the tools presented in this paper for the satisfaction of basic formulas having one operator
can be applied to certify  $f$.
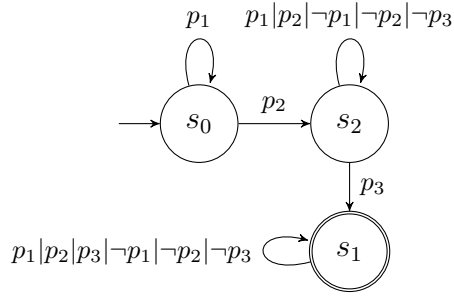
\begin{figure}[!t]
\centering
\begin{tikzpicture}
	\node[state, initial, black] (s0) at (0,0) {$s_0$};
	\node[state, black] (s2) at (2,0) {$s_2$};
	\node[state, accepting] (s1) at (2,-1.7) {$s_1$};
	\draw	(s0) edge[loop above, black] node{\footnotesize$p_1$} (s0)
			(s2) edge[loop above, black] node{\footnotesize$p_1|p_2|\lnot p_1|\lnot p_2|\lnot p_3$} (s2)
			(s0) edge[above, black] node{\footnotesize$p_2$} (s2)
			(s1) edge[loop left] node{\footnotesize$p_1|p_2|p_3|\lnot p_1|\lnot p_2|\lnot p_3$} (s1)
			(s2) edge[right] node{\footnotesize$p_3$} (s1);
\end{tikzpicture}
\label{fig:fsa_combining}
\caption{An example of an FSA representing the formula $f = \Diamond p_3 \land (p_1 \,\mathcal{U}_s p_2)$. The state $s_0$ is the initial state and $s_1$ is the final state. When several transitions are present between two states, one transition labeled by the set of all observations using the symbol $|$ as shown.}
\label{fig:fsa_combining}
\end{figure}
For instance, 
the formula $f = \Diamond p_3 \land (p_1 \,\mathcal{U}_s p_2)$ has the
following associated FSA:
$A_f = (S, s_0, O, \delta, S_F)$, where
\begin{equation}
\begin{split}
 &S = \{s_0, s_1, s_2\}, \quad S_F = \{s_1\}, \quad O = \{p_1, p_2, p_3, \lnot p_1, \lnot p_2, \lnot p_3\},\\
 &\delta(s,o) =
 \left\{
\begin{array}{ll}
  s_0 & \mbox{if } s = s_0\mbox{, } o = p_1\mbox{, }\\
  s_2 & \mbox{if } s = s_0\mbox{, } o = p_2\mbox{, }\\
  s_2 & \mbox{if } s = s_2\mbox{, } o \neq p_3\mbox{, }\\
  s_1 & \mbox{if } s = s_2\mbox{, } o = p_3\mbox{, }\\
  s_1 & \mbox{if } s = s_1.\end{array}\right.
 \quad \forall (s,o) \in S \times O
\end{split}
\label{eqn:Af_ex}
\end{equation}
This FSA is shown in Figure~\ref{fig:fsa_combining}.
As shown therein, the FSA state $s$ is initially at $s_0$ and
when $s$ reaches the final state $s_1$, it implies that the given formula $f$ is satisfied.
As $s$ starts at $s_0$, we must have that the initial observation $o$ is either $o=p_1$ or $o=p_2$. If it is $o=p_1$, $s$ remains at $s_0$, but if $o=p_2$, we have a transition from $s=s_0$ to $s=s_2$.
Then, once $s$ is at $s_2$, we have a transition of $s$ from $s_2$ to $s_1$ if $o = p_3$. If $o \neq p_3$, $s$ remains at $s_2$.
 In other words, the FSA captures the given formula as follows:
\begin{itemize}
 \item[1)] When $s$ is at $s_0$, $p_2$ has to be eventually satisfied and $p_1$ has to be satisfied until $p_2$ is satisfied; i.e.,~$p_1 \mathcal{U}_s p_2$ is satisfied. Once $p_2$ is satisfied, $s$ jumps to $s_2$.
 \item[2)] When $s$ is at $s_2$, $p_3$ needs to be eventually satisfied for $f$ to be satisfied; i.e.,~$\Diamond p_3$ is satisfied. Additionally, once $p_3$ is satisfied, $s$ jumps to $s_1$.
\end{itemize}

To apply our tools, 
by extending the ideas in~\cite{bisoffi2018hybrid}, 
we build an augmented version of $\mathcal{H}$, denoted by $\mathcal{H}_A$, with state $(x,s) \in \mathcal{X} \times S$ and input $o \in O$ in which $s$ transitions according to the FSA associated with the formula.  
 Its input $o$, namely, the observation $o$, is determined by 
the propositions that are satisfied (or not). For example, when 
$x$ is such that $p_1(x) = 1$ then $o = p_1$, 
while when $p_1(x) = 0$ then $o = \lnot p_1$.
Then, according to our tools, the satisfaction of the formula $f$ is assured by the following conditions:
\begin{itemize}
 \item Conditions in Theorem~\ref{thm:until_2}, with $q$ therein replaced by $p_2$ and with $P=\{(x,s) \in \mathcal{X} \times S: p_1(x)=1, s=s_0\}$ and $Q=\{(x,s) \in \mathcal{X} \times S: p_2(x)=1, s=s_2\}$, are satisfied; and
 \item $K = \{(x,s) \in \mathcal{X} \times S: p_3(x) = 1, s = s_1\}$ is FTA for $\mathcal{H}_A$; namely, conditions in Theorem~\ref{thm:FTA}, with $p$ therein replaced by $p_3$ and with set $K$ just defined, are satisfied.
\end{itemize}
The methodology outlined above can be automated, and is part of current research.
 \section{Conclusion}

When the hybrid system is well-posed as defined in~\cite{goebel2012hybrid},
which requires mild conditions on the system data, the satisfaction of the
formulas guaranteed by our results are robust to small general perturbations
over finite hybrid time horizons. Such intrinsic robustness, though small, is
useful in real-world applications as it allows for small errors on the initial
conditions, small perturbations during flows (both on $F$ and
$C$) and at jumps (both on $G$ and $D$).

Future work includes certifying formulas for hybrid systems with robustness to large perturbations and over the infinite horizon.
The authors in~\cite{fainekos2009robustness} propose robust semantics for
Metric Temporal Logic (MTL) formulas to prove that a continuous-time signal
satisfies an MTL specification robustly.
By extending the ideas therein to our setting, it might be possible to assure the satisfaction of temporal logic specifications with robustness to large perturbations.

\section*{Acknowledgement}
This research has been partially supported by the National Science Foundation under Grant no. ECS-1710621 and Grant no. CNS-1544396, by the Air Force Office of Scientific Research under Grant no. FA9550-16-1-0015, Grant no. FA9550-19-1-0053, and Grant no. FA9550-19-1-0169, and by CITRIS and the Banatao Institute at the University of California.
\appendix
\section*{Appendix}

\renewcommand{\thesection}{A} 
\section{Nonsmooth Lyapunov Functions}
\label{appendix:nonsmooth_ly}

For a hybrid system $\mathcal{H}=(C,F,D,G)$, let $V : \mathcal{X} \rightarrow \mathbb{R}$ be continuous on $\mathcal{X}$ and locally Lipschitz on a neighborhood of $C$. The generalized gradient of $V$ at $x \in C$, denoted by $\partial V(x)$, is a closed, convex and nonempty set equal to the convex hull of all limits of the sequence $\nabla V(x_i)$, where $x_i$ is any sequence converging to $x$ while $x$ avoids an arbitrary set of measure zero containing all the points at which $V$ is not differentiable. As $V$ is locally Lipschitz, $\nabla V$ exists almost everywhere.
The generalized directional derivative of $V$ at $x$ in the direction of $v$ can be presented as follows \cite{clarke1990optimization}:
\begin{equation}
	V^\circ (x,v) = \max_{\zeta \in \partial V(x)} \langle \zeta,v \rangle \mbox{.}
\end{equation}
In addition, for any solution $t \mapsto \phi(t,0)$ to $\dot{\phi}(t,0) \in F(x)$,
\begin{equation}
	\tfrac{d}{dt} V(\phi(t,0)) \leq V^\circ (\phi(t,0), \dot{\phi}(t,0))
\end{equation}
for almost all $t$ in the domain of definition of $\phi$, where $\tfrac{d}{dt} V(\phi(t,0))$ is understood in the standard sense since $V$ is locally Lipschitz.
%
%

To bound the increase of the function $V$ along solutions to a hybrid system $\mathcal{H}$, we define the function $u_C : \mathcal{X} \rightarrow [-\infty, +\infty)$ as follows \cite{18}:
\begin{equation}
	u_C (x) :=
\left\{
	\begin{array}{ll}
		\underset{{v \in F(x)}}{\max} \: \underset{\zeta \in \partial V(x)}{\max} \langle \zeta,v \rangle & x \in C\\
		-\infty & \mbox{otherwise.}
	\end{array}
\right.
\label{eqn:uc_1}
\end{equation}
In particular, for any solution $\phi$ to $\mathcal{H}$ and any $t$ where $\tfrac{d}{dt} V(\phi(t,j))$ exists, we have
\begin{equation}
	\tfrac{d}{dt} V(\phi(t,j)) \leq u_C (\phi(t,j))\mbox{.}
\label{eqn:uc_2}
\end{equation}
%

Furthermore, in order to bound the change in $V$ after jumps, we define the following quantity:
\begin{equation}
	u_D(x) :=
\left\{
	\begin{array}{ll}
		\underset{\zeta \in G(x)}{\max} V(\zeta) - V(x) & x \in D\\
		-\infty & \mbox{otherwise.}
	\end{array}
\right.
\label{eqn:ud_1}
\end{equation}
Then, for any solution $\phi$ to $\mathcal{H}$ and for any $(t_{j+1}, j), (t_{j+1}, j+1) \in$ $\dom\phi$, it follows that
\begin{equation}
	V(\phi(t_{j+1}, j+1)) - V(\phi(t_{j+1}, j)) \leq u_D (\phi (t_{j+1}, j)) \mbox{.}
\label{eqn:ud_2}
\end{equation}
%

Note that when $F$ is a single-valued map, $u_C(x) = V^\circ (x, F(x))$ for each $x \in C$. When $G$ is a single-valued map, $u_D(x) = V(G(x)) - V(x)$ for each $x \in D$.

\renewcommand{\thesection}{B}
\renewcommand{\thesection}{B}
\section{Results on Finite Time Attractivity}
\label{appendix:eventually}

In the following, we present sufficient conditions that guarantee FTA of a closed set $K$ for a hybrid system $\mathcal{H}$; see \cite{188}. First, Proposition \ref{prop:FTA_flow} characterizes the scenario where the distance of each solution $\phi \in \mathcal{S_H} (\mathcal{O})$ to $K$ strictly decreases during flows, but is only non-increasing at jumps, where $\mathcal{O} \subset \mathcal{N}$ and $\mathcal{N}$ is an open neighborhood of $K$.

\begin{proposition}
Let a hybrid system $\mathcal{H}=(C,F,D,G)$ on $\mathcal{X}$ and a closed set $K \!\subset\! \mathcal{N} \!\subset\! \mathcal{X}$ with an open set $\mathcal{N}$ such that $G(\mathcal{N}) \!\subset\! \mathcal{N}$.
If there exists a continuous function $V:\mathcal{N} \rightarrow \mathbb{R}_{\geq 0}$, locally Lipschitz on an open neighborhood of $C \cap \mathcal{N}$, and $c_1 > 0$, $c_2 \in [0,1)$ such that
	\begin{itemize}
	\item[1)] for every $x \in \mathcal{N} \cap (\overline{C} \cup D) \setminus K$, each $\phi \in \mathcal{S_H} (x)$ satisfies
	\[
		\tfrac{V^{1-c_2} (x)}{c_1 (1-c_2)} \leq \sup_{(t,j) \in \textrm{dom}\,\phi} t \mbox{;}
	\]
	\item[2)] the function $V$ is positive definite with respect to $K$ and
	\begin{subequations}
	\begin{align}
		u_C(x) + c_1 V^{c_2} (x) \leq 0 \phantom{u_D(x)} & \qquad\forall x \in (C \cap \mathcal{N}) \setminus K\mbox{,} \label{eqn:flow_uc}\\
		u_D(x) \leq 0 \phantom{u_C(x) + c_1 V^{c_2} (x)} & \qquad\forall x \in (D \cap \mathcal{N}) \setminus K \label{eqn:jump_ud} \mbox{,}
	\end{align}
	\label{eqn:all_uc_ud}
	\end{subequations}
	where the functions $u_C$ and $u_D$ are defined in \eqref{eqn:uc_1} and \eqref{eqn:ud_1}, respectively.
	\end{itemize}
Then, the set $K$ is FTA for $\mathcal{H}$ with respect to $\mathcal{O} := L_V(r) \cap (\overline{C} \cup D)$ where $L_V(r) \!=\! \{x \!\in\! \mathcal{X} : V(x) \leq r \}, r \!\in\! [0, \infty]$, is a sublevel set of $V$ contained in $\mathcal{N}$.

Furthermore, for each $\phi \in \mathcal{S_H} (\mathcal{O})$,
defining $\xi := \phi(0,0)$, the following holds:
\begin{itemize}
	\item[a)] The settling-time function $\mathcal{T}$ satisfies $\mathcal{T}(\phi) \leq \mathcal{T}^\star(\xi) + \mathcal{J}^\star(\phi)$
	where $\mathcal{T}^\star(\xi) = \tfrac{V^{1-c_2} (\xi)}{c_1 (1-c_2)}$ and $\mathcal{J}^\star(\phi)$ is such that $(\mathcal{T}^\star(\xi), \mathcal{J}^\star(\phi)) \in \dom\phi$; and
	\item[b)] $|\phi(t,j)|_K = 0$ for some $(t,j) \in \dom\phi$ such that $t \leq \mathcal{T}^\star(\xi)$.
\end{itemize}
\label{prop:FTA_flow}
\end{proposition}%
\begin{proof}
Let $\phi \in \mathcal{S_H} (\xi)$ with $\xi \in L_V(r) \cap (\overline{C} \cup D)$.
Pick any $(t,j) \in$ $\dom \phi$ and let $0 = t_0 \leq t_1 \leq \cdots \leq t_{j+1} = t$ satisfy
\begin{equation}
	\dom \phi \cap ( [0,t] \times \{0,1,\dots,j\} ) = \bigcup_{i=0}^j([t_i, t_{i+1}] \times \{i\}) \mbox{.}
\end{equation}
Now, suppose that, for each $i \in \{0,1,\dots,j\}$ and almost all $s \in [t_i, t_{i+1}]$, $\phi(s,i) \in (C \cap L_V(r)) \setminus K$.
We will later show that this is the case by picking $t = t_{j+1}$ and $j$ appropriately.
Note that $L_V(r) \subset \mathcal{N}$ by assumption.
Using \eqref{eqn:uc_2}, the condition in \eqref{eqn:flow_uc} implies that, for each $i \in \{0,1,\dots,j\}$ and for almost all $s \in [t_i, t_{i+1}]$,
\begin{equation}
	\tfrac{d}{ds} V(\phi(s,i)) \leq u_C(\phi(s,i)) \leq -c_1 V^{c_2}(\phi(s,i)) \mbox{,}
\end{equation}
which implies that
\begin{equation}
	V^{-c_2} (\phi(s,i)) \,dV (\phi(s,i)) \leq -c_1 ds \mbox{.}
\end{equation}
Integrating over $[t_i, t_{i+1}]$ both sides of this inequality yields
\begin{equation}
	\tfrac{1}{1-c_2} \left(V^{1-c_2}(\phi(t_{i+1}, i)) - V^{1-c_2} (\phi(t_i,i))\right) \leq -c_1(t_{i+1} - t_i) \mbox{.}
\label{eqn:proof_c}
\end{equation}
Similarly, for each $i \in \{1,\dots,j\}$, $\phi(t_i,i-1) \in (D \cap L_V(r))\setminus K$.
As stated above, we will later show that this is the case by picking $t = t_{j+1}$ and $j$ appropriately.
Then,
\begin{equation}
	V(\phi(t_i,i)) - V(\phi(t_i, i-1)) \leq 0 \mbox{.}
\label{eqn:proof_d}
\end{equation}
The two inequalities in \eqref{eqn:proof_c} and \eqref{eqn:proof_d} imply that, for each $(t,j) \in \dom \phi$,
\begin{equation}
	\tfrac{1}{1-c_2} \left(V^{1-c_2}(\phi(t,j)) - V^{1-c_2} (\xi) \right) \leq -c_1 t \mbox{.}
\end{equation}
Using the fact that $c_2 \in (0,1)$, we get 
\begin{equation}
	V^{1-c_2} (\phi(t,j)) \leq V^{1-c_2}(\xi) - c_1 (1-c_2) t \mbox{.}
\end{equation}
Then, it follows that the quantity $V^{1 - c_2}(\xi) - c_1(1 - c_2)t$ converges to zero in finite time that is upper bounded by $\tfrac{V^{1 - c_2}(\xi)}{c_1(1 - c_2)}$,
which, in turn, implies that 
$\phi$ converges to the set $K$ in finite time since $V$ is positive definite with respect to $K$ and using item 1). Indeed, item 1) implies that each solution starting from 
$\xi$ satisfies $\dom \phi \cap 
\big( \big\{ \tfrac{V^{1 - c_2}(\xi)}{c_1(1 - c_2)} \big\} \times \mathbb{N}\big) \neq \emptyset$.
Consequently, we conclude that there exists a hybrid time $(t^\star,j^\star) \in \dom\phi$ such that $\phi(t^\star,j^\star) \in K$ with $t^\star \leq \tfrac{V^{1 - c_2}(\xi)}{c_1(1 - c_2)}$.

Next, we show that $\phi$ stays in $L_V(r)$ until it reaches the set $K$.
Proceeding by contradiction, if $\phi$ does not stay in $L_V(r)$ until $\phi$ reaches the set $K$, then there exists a first hybrid time $(t',j') \in \dom\phi$ such that $V(\phi(t',j')) > r$ and $\phi(t,j) \notin K$ for every $(t,j) \in \dom\phi$ such $t + j \leq t' + j'$. Note that $\phi(0,0) \in L_V(r) \subset \mathcal{N}$. Then, using the conditions in \eqref{eqn:all_uc_ud}, the fact that $G(\mathcal{N}) \subset \mathcal{N}$, the continuity of $V$ and the concept of solutions to hybrid inclusions, we conclude that $V(\phi(t',j')) \leq V(\phi(0,0)) \leq r$ since $\phi(t,j) \in (\overline{C} \cup D) \setminus K$ for every $(t,j) \in \dom\phi$ such $t + j \leq t' + j'$. Hence, the contradiction follows.

Furthermore, an upper bound for the settling-time function can be computed as
\begin{equation}
	\mathcal{T}(\phi) \leq \mathcal{T}^\star (\xi) + \mathcal{J}^\star (\phi) \mbox{,}
\end{equation}
where $\mathcal{T}^\star(\xi) = \frac{V^{1-c_2} (\xi)}{c_1 (1-c_2)}$, and $\mathcal{J}^\star(\phi)$ can be  chosen as
$$\mathcal{J}^\star(\phi) = \sup_{\substack{(t,j)\in\dom\phi\\ t < \mathcal{T}^\star(\xi)}} j\mbox{.}$$
Note that since $\mathcal{T}^\star(\xi) \leq \sup_{(t,j)\in \dom \phi} t$, the existence of $(\mathcal{T}^\star (\xi), \mathcal{J}^\star (\phi)) \in \dom \phi$ is guaranteed.
\end{proof}

\begin{remark}
Condition 1) in Proposition \ref{prop:FTA_flow} guarantees that the domain of definition of the solutions to $\mathcal{H}$ are long enough to allow for the solution to converge to $K$. Condition \eqref{eqn:flow_uc} guarantees finite time convergence of $\lim_{t+j \to \mathcal{T}(\phi)} |\phi(t,j)|_K$ to zero over a finite amount of ordinary time $t$ (potentially with jumps within it). Finally, the upper bound on the settling-time function $\mathcal{T}$ depending on the Lyapunov function and the initial condition will be effectively exploited to estimate the amount of hybrid time it takes for a temporal specification to be satisfied.
\end{remark}

A dual version of Proposition \ref{prop:FTA_flow} is given next, namely, it pertains to the case,
when the distance of a solution $\phi \in \mathcal{S_H}(\mathcal{O})$ to a closed set $K$ strictly decreases at jumps where $\mathcal{O} \subset \mathcal{N}$ and $\mathcal{N}$ is an open neighborhood of $K$.

\begin{proposition}
Let a hybrid system $\mathcal{H}=(C,F,D,G)$ on $\mathcal{X}$ and a closed set $K \!\subset\! \mathcal{N} \!\subset\! \mathcal{X}$ with an open set $\mathcal{N}$ such that $G(\mathcal{N}) \!\subset\! \mathcal{N}$.
If there exists a continuous function $V:\mathcal{N} \rightarrow \mathbb{R}_{\geq 0}$, locally Lipschitz on an open neighborhood of $C \cap \mathcal{N}$, and $c > 0$ such that
\begin{itemize}
	\item[1)] for every $x \in \mathcal{N} \cap (\overline{C} \cup D) \setminus K$, each $\phi \in \mathcal{S_H}(x)$ satisfies
		\[
			\mbox{ceil} \left( \tfrac{V(x)}{c} \right) \leq \sup_{(t,j) \in \textrm{dom}\,\phi} j \mbox{;}
		\]
	\item[2)] the function $V$ is positive definite with respect to $K$ and
\begin{subequations}
	\begin{align}
		u_C(x) \leq 0\phantom{-\min\{c, V(x)\}} &\qquad\forall x \in (C \cap \mathcal{N}) \setminus K\mbox{,} \label{eqn:jump_uc_2}\\
		u_D(x) \leq -\min\{c, V(x)\}\phantom{0}& \qquad\forall x \in (D \cap \mathcal{N}) \setminus K \label{eqn:jump_ud_2} \mbox{,}
	\end{align}
\label{eqn:all_uc_ud_2}
\end{subequations}
	where $u_C$ and $u_D$ are defined in \eqref{eqn:uc_1} and \eqref{eqn:ud_1}, respectively.
\end{itemize}
Then, the set $K$ is FTA for $\mathcal{H}$ with respect to $\mathcal{O} := L_V(r) \cap (\overline{C} \cup D)$ where $L_V(r) \!=\! \{x \!\in\! \mathcal{X} : V(x) \leq r \}, r \!\in\! [0, \infty]$, is a sublevel set of $V$ contained in $\mathcal{N}$.

Furthermore, for each $\phi \in \mathcal{S_H} (\mathcal{O})$, defining $\xi = \phi(0,0)$, the following holds:
\begin{itemize}
	\item[a)] the settling-time function $\mathcal{T}$ satisfies $\mathcal{T}(\phi) \leq \mathcal{T}^\star(\phi) + \mathcal{J}^\star(\xi)$
	where $\mathcal{J}^\star(\xi) = \mbox{ceil} \left( \tfrac{V(\xi)}{c} \right)$ and $\mathcal{T}^\star (\phi)$ is such that $(\mathcal{T}^\star(\phi), \mathcal{J}^\star(\xi)) \in \dom\phi$ and $(\mathcal{T}^\star(\phi), \mathcal{J}^\star(\xi)-1) \in \dom\phi$;
	\item[b)] $|\phi(t,j)|_K = 0$ for some $(t,j) \in \dom\phi$ such that $j \leq \mathcal{J}^\star(\xi)$.
\end{itemize}
\label{prop:FTA_jump}
 
\begin{proof}
Let $\phi \in \mathcal{S_H}(\xi)$ with $\xi \in L_V(r) \cap (\overline{C} \cup D)$.
Pick any $(t,j) \in \dom \phi$ and let $0 = t_0 \leq t_1 \leq \dots \leq t_{j+1} = t$ satisfy
\begin{equation}
	\dom \phi \cap ( [0,t] \times \{0,1,\dots,j\} ) = \bigcup_{i=0}^j([t_i, t_{i+1}] \times \{i\}) \mbox{.}
\end{equation}
Now, suppose that, for each $i \in \{0,1,\dots,j\}$ and almost all $s \in [t_i, t_{i+1}]$, $\phi(s,i) \in (C \cap L_V(r)) \setminus K$.
We will later show that this is the case by picking $t = t_{j+1}$ and $j$ appropriately.
Note that $L_V(r) \subset \mathcal{N}$ by assumption.
Using \eqref{eqn:uc_2}, the condition in \eqref{eqn:jump_uc_2} implies that, for each $i \in \{0,1,\dots,j\}$ and for almost all $s \in [t_i,t_{i+1}]$,
$\tfrac{dV(\phi(s,i))}{ds} \leq 0$.
Integrating over $[t_i, t_{i+1}]$ both sides of this inequality yields
\begin{equation}
	V(\phi(t_{i+1},i)) - V(\phi(t_i,i)) \leq 0 \mbox{.}
\label{eqn:proof_fta1}
\end{equation}
Similarly, by using \eqref{eqn:ud_2} and \eqref{eqn:jump_ud_2}, for each $i \in \{1, \dots, j\}$, $\phi(t_i, i-1) \in (D \cap L_V(r)) \setminus K$ and
\begin{equation}
	V (\phi(t_i,i)) - V(\phi(t_i,i-1)) \leq - \min \{c, V(\phi(t_i,i-1))\} \mbox{.}
\label{eqn:proof_fta2}
\end{equation}
The two inequalities in \eqref{eqn:proof_fta1} and \eqref{eqn:proof_fta2} imply that, for each $(t,j) \in \dom\phi$,
\begin{equation*}
	V (\phi(t,j)) - V(\xi) \leq -\sum_{i=1}^j \min \{c, V(\phi(t_i,i-1))\} \mbox{.}
\end{equation*}
Then, it follows that the quantity $V(\xi) - \sum_{i=1}^j \min \{c, V(\phi(t_i,i-1))$ converges to zero in finite time that is upper bounded by $\mbox{ceil} \big( \tfrac{V(\xi)}{c} \big)$, which implies that $\phi$ converges to the set $K$ in finite time since $V$ is positive definite with respect to $K$ and due to item 1).
Indeed, item 1) implies that each solution starting from $\xi$ satisfies $\dom\phi \cap ( \mathbb{R}_{\geq 0} \times \{0,1,\dots, J\} ) \neq \emptyset$ where $J = \mbox{ceil} \big( \tfrac{V(\xi)}{c} \big)$.
Consequently, we conclude that there exists a hybrid time $(t^\star, j^\star) \in \dom\phi$ such that $\phi(t^\star, j^\star) \in K$ with $j^\star \leq \mbox{ceil} \big( \tfrac{V(\xi)}{c} \big)$.

Furthermore, an upper bound for the settling-time function can be computed as
\begin{equation}
	\mathcal{T}(\phi) \leq \mathcal{T}^\star (\phi) + \mathcal{J}^\star (\xi) \mbox{,}
\end{equation}
where $\mathcal{J}^\star(\xi) = \mbox{ceil} \left( \tfrac{V(\xi)}{c} \right)$ and $\mathcal{T}^\star (\phi)$ is such that $(\mathcal{T}^\star(\phi), \mathcal{J}^\star(\xi)), (\mathcal{T}^\star(\phi), \mathcal{J}^\star(\xi) - 1) \in \dom \phi$.
Note that since $\mathcal{J}^\star(\xi) < \sup_{(t,j)\in \dom \phi} j$, the existence of $(\mathcal{T}^\star (\phi), \mathcal{J}^\star (\xi)) \in \dom \phi$ is guaranteed.
\end{proof}
\end{proposition}

The following result combines the conditions in Proposition \ref{prop:FTA_flow} and in Proposition \ref{prop:FTA_jump}.
Its proof can be formulated by combining the arguments in the proofs of Proposition \ref{prop:FTA_flow} and Proposition \ref{prop:FTA_jump}.

\begin{proposition}
Let a hybrid system $\mathcal{H}=(C,F,D,G)$ on $\mathcal{X}$ and a closed set $K \!\subset\! \mathcal{N} \!\subset\! \mathcal{X}$ with an open set $\mathcal{N}$ such that $G(\mathcal{N}) \!\subset\! \mathcal{N}$. If there exists a continuous function $V:\mathcal{N} \rightarrow \mathbb{R}_{\geq 0}$, locally Lipschitz on an open neighborhood of $C \cap \mathcal{N}$, and $c_1, c_3 > 0$, $c_2 \in [0,1)$ such that item 1) in Proposition \ref{prop:FTA_flow} and item 1) in Proposition \ref{prop:FTA_jump} are satisfied,
the function $V$ is positive definite with respect to $K$ and
\begin{subequations}
	\begin{align}
		u_C(x) \leq - c_1 V^{c_2}(x) \phantom{-\min\{c_3, V(x)\}} & \forall x \in (C \cap \mathcal{N}) \setminus K\mbox{,} \label{eqn:combine_uc}\\
		u_D(x) \leq -\min\{c_3, V(x)\} \phantom{-c_1 V^{c_2}(x)} &\forall x \in (D \cap \mathcal{N}) \setminus K \label{eqn:combine_ud} \mbox{,}
	\end{align}
	\label{eqn:all_uc_ud_3}
\end{subequations}
where $u_C$ and $u_D$ are defined in \eqref{eqn:uc_1} and \eqref{eqn:ud_1}, respectively.
Then, the set $K$ is FTA for $\mathcal{H}$ with respect to $\mathcal{O} := L_V(r) \cap (\overline{C} \cup D)$ where $L_V(r) \!=\! \{x \!\in\! \mathcal{X} : V(x) \leq r \}, r \!\in\! [0, \infty]$, is a sublevel set of $V$ contained in $\mathcal{N}$.

Furthermore, for each $\phi \in \mathcal{S_H} (\mathcal{O})$,
the following holds:
\begin{itemize}
	\item[a)] the settling-time function $\mathcal{T}$ satisfies $\mathcal{T}(\phi) \leq \min_{i \in \{1,2\}} \left\{\mathcal{T}_i^\star(\phi) + \mathcal{J}_i^\star(\phi)\right\}$
	where $\mathcal{T}_1^\star(\phi) = \tfrac{V^{1-c_2} (\phi(0,0))}{c_1(1-c_2)}$, $\mathcal{J}_1^\star(\phi)$ is such that $(\mathcal{T}_1^\star(\phi), \mathcal{J}_1^\star(\phi)) \in \dom\phi$, $\mathcal{J}_2^\star(\phi)= \mbox{ceil} \Big( \tfrac{V(\phi(0,0))}{c_3} \Big)$, and $\mathcal{T}_2^\star(\phi)$ is such that $(\mathcal{T}_2^\star(\phi), \mathcal{J}_2^\star(\phi)) \in \dom \phi$ and $(\mathcal{T}_2^\star(\phi), \mathcal{J}_2^\star(\phi)-1) \in \dom\phi$;
	\item[b)] $|\phi(t,j)|_K = 0$ for some $(t,j) \in \dom\phi$ such that $t \leq \mathcal{T}_1^\star(\phi)$ or $j \leq \mathcal{J}_2^\star(\phi)$. 
\end{itemize}
\label{prop:FTA_flow_jump}
\end{proposition}

\bibliographystyle{IEEEtran}
\bibliography{LTL_NAHS_arXiv}

\begin{thebibliography}{10}
\providecommand{\url}[1]{#1}
\csname url@rmstyle\endcsname
\providecommand{\newblock}{\relax}
\providecommand{\bibinfo}[2]{#2}
\providecommand\BIBentrySTDinterwordspacing{\spaceskip=0pt\relax}
\providecommand\BIBentryALTinterwordstretchfactor{4}
\providecommand\BIBentryALTinterwordspacing{\spaceskip=\fontdimen2\font plus
\BIBentryALTinterwordstretchfactor\fontdimen3\font minus
  \fontdimen4\font\relax}
\providecommand\BIBforeignlanguage[2]{{%
\expandafter\ifx\csname l@#1\endcsname\relax
\typeout{** WARNING: IEEEtran.bst: No hyphenation pattern has been}%
\typeout{** loaded for the language `#1'. Using the pattern for}%
\typeout{** the default language instead.}%
\else
\language=\csname l@#1\endcsname
\fi
#2}}

\bibitem{tabuada2006linear}
P.~Tabuada and G.~J. Pappas, ``Linear time logic control of discrete-time
  linear systems,'' \emph{IEEE Transactions on Automatic Control}, vol.~51,
  no.~12, pp. 1862--1877, 2006.

\bibitem{kloetzer2008fully}
M.~Kloetzer and C.~Belta, ``A fully automated framework for control of linear
  systems from temporal logic specifications,'' \emph{IEEE Transactions on
  Automatic Control}, vol.~53, no.~1, pp. 287--297, 2008.

\bibitem{kwon2008ltlc}
Y.~Kwon and G.~Agha, ``{LTLC}: Linear temporal logic for control,''
  \emph{Hybrid Systems: Computation and Control}, pp. 316--329, 2008.

\bibitem{Pnueli.77}
A.~Pnueli, ``The temporal logic of programs,'' in \emph{18th Annual Symposium
  on Foundations of Computer Science, 1977}.\hskip 1em plus 0.5em minus
  0.4em\relax IEEE, 1977, pp. 46--57.

\bibitem{MannaPnueli.92}
A.~Pnueli and Z.~Manna, ``The temporal logic of reactive and concurrent
  systems,'' \emph{Springer}, vol.~16, p.~12, 1992.

\bibitem{Fainekos.ea.09.Automatica}
G.~E. Fainekos, A.~Girard, H.~Kress-Gazit, and G.~J. Pappas, ``Temporal logic
  motion planning for dynamic robots,'' \emph{Automatica}, vol.~45, no.~2, pp.
  343--352, 2009.

\bibitem{karaman2008optimal}
S.~Karaman, R.~G. Sanfelice, and E.~Frazzoli, ``Optimal control of mixed
  logical dynamical systems with linear temporal logic specifications,'' in
  \emph{Decision and Control, 2008. CDC 2008. 47th IEEE Conference on}.\hskip
  1em plus 0.5em minus 0.4em\relax IEEE, 2008, pp. 2117--2122.

\bibitem{Wolff.ea.14.ICRA}
E.~M. Wolff, U.~Topcu, and R.~M. Murray, ``Optimization-based trajectory
  generation with linear temporal logic specifications,'' in \emph{Robotics and
  Automation (ICRA), 2014}.\hskip 1em plus 0.5em minus 0.4em\relax IEEE, 2014,
  pp. 5319--5325.

\bibitem{dimitrova2014deductive}
R.~Dimitrova and R.~Majumdar, ``Deductive control synthesis for
  alternating-time logics,'' in \emph{Proceedings of the 14th International
  Conference on Embedded Software}.\hskip 1em plus 0.5em minus 0.4em\relax ACM,
  2014, p.~14.

\bibitem{Saha.Julius.16.ACC}
S.~Saha and A.~A. Julius, ``An {MILP} approach for real-time optimal controller
  synthesis with metric temporal logic specifications,'' in \emph{American
  Control Conference (ACC), 2016}.\hskip 1em plus 0.5em minus 0.4em\relax IEEE,
  2016, pp. 1105--1110.

\bibitem{bisoffi2018hybrid}
A.~Bisoffi and D.~V. Dimarogonas, ``A hybrid barrier certificate approach to
  satisfy linear temporal logic specifications,'' in \emph{2018 Annual American
  Control Conference (ACC)}.\hskip 1em plus 0.5em minus 0.4em\relax IEEE, 2018,
  pp. 634--639.

\bibitem{Raman.ea.15.HSCC}
V.~Raman, A.~Donz{\'e}, D.~Sadigh, R.~M. Murray, and S.~A. Seshia, ``Reactive
  synthesis from signal temporal logic specifications,'' in \emph{Proceedings
  of the 18th International Conference on Hybrid Systems: Computation and
  Control}.\hskip 1em plus 0.5em minus 0.4em\relax ACM, 2015, pp. 239--248.

\bibitem{vanderSchaftSchumacher00}
A.~van~der Schaft and H.~Schumacher, \emph{An Introduction to Hybrid Dynamical
  Systems}.\hskip 1em plus 0.5em minus 0.4em\relax Lecture Notes in Control and
  Information Sciences, Springer, 2000.

\bibitem{LygerosJohanssonSimicZhangSastry03}
J.~Lygeros, K.~H. Johansson, S.~N. Simic, J.~Zhang, and S.~S. Sastry,
  ``Dynamical properties of hybrid automata,'' \emph{IEEE Transactions on
  automatic control}, vol.~48, no.~1, pp. 2--17, 2003.

\bibitem{Collins04}
P.~Collins, ``A trajectory-space approach to hybrid systems,'' in
  \emph{MTNS2004z}, 2004.

\bibitem{HaddadChellaboinaNersesov06}
W.~M. Haddad, V.~Chellaboina, and S.~G. Nersesov, \emph{Impulsive and Hybrid
  Dynamical Systems: Stability, Dissipativity, and Control}.\hskip 1em plus
  0.5em minus 0.4em\relax Princeton University, 2006.

\bibitem{goebel2012hybrid}
R.~Goebel, R.~G. Sanfelice, and A.~R. Teel, \emph{Hybrid Dynamical Systems:
  modeling, stability, and robustness}.\hskip 1em plus 0.5em minus 0.4em\relax
  Princeton University Press, 2012.

\bibitem{Johansson99}
K.~Johansson, M.~Egerstedt, J.~Lygeros, and S.~Sastry, ``On the regularization
  of {Z}eno hybrid automata,'' \emph{Systems \& Control Letters}, vol.~38,
  no.~3, pp. 141--150, 1999.

\bibitem{fainekos2009robustness}
G.~E. Fainekos and G.~J. Pappas, ``Robustness of temporal logic specifications
  for continuous-time signals,'' \emph{Theoretical Computer Science}, vol. 410,
  no.~42, pp. 4262--4291, 2009.

\bibitem{185}
J.~Chai and R.~G. Sanfelice, ``Forward invariance of sets for hybrid dynamical
  systems {(Part I)},'' \emph{IEEE Transactions on Automatic Control}, 2018.

\bibitem{188}
Y.~Li and R.~G. Sanfelice, ``Finite time stability of sets for hybrid dynamical
  systems,'' \emph{Automatica}, vol. 100, pp. 200--211, 2019.

\bibitem{cimatti2015hreltl}
A.~Cimatti, M.~Roveri, and S.~Tonetta, ``{HRELTL}: A temporal logic for hybrid
  systems,'' \emph{Information and Computation}, vol. 245, pp. 54--71, 2015.

\bibitem{Khalil}
H.~Khalil, \emph{Nonlinear Systems}, 3rd~ed.\hskip 1em plus 0.5em minus
  0.4em\relax Prentice Hall, 2002.

\bibitem{176}
H.~Han and R.~G. Sanfelice, ``Sufficient conditions for temporal logic
  specifications in hybrid dynamical systems,'' in \emph{Proceedings of the 6th
  Analysis and Design of Hybrid Systems}, vol. Volume 51, 2018, pp. 97--102.

\bibitem{105}
R.~G. Sanfelice, \emph{Analysis and Design of Cyber-Physical Systems: A Hybrid
  Control Systems Approach}.\hskip 1em plus 0.5em minus 0.4em\relax CRC Press,
  2015, pp. 3--31.

\bibitem{Maghenem.Sanfelice.18.}
M.~Maghenem and R.~G. Sanfelice, ``Barrier function certificates for invariance
  in hybrid inclusions,'' in \emph{Decision and Control (CDC), 57th Annual
  Conference on}.\hskip 1em plus 0.5em minus 0.4em\relax IEEE, 2018.

\bibitem{subbaraman2016equivalence}
A.~Subbaraman and A.~R. Teel, ``On the equivalence between global recurrence
  and the existence of a smooth lyapunov function for hybrid systems,''
  \emph{Systems \& Control Letters}, vol.~88, pp. 54--61, 2016.

\bibitem{goebel2009hybrid}
R.~Goebel, R.~G. Sanfelice, and A.~R. Teel, ``Hybrid dynamical systems,''
  \emph{IEEE Control Systems}, vol.~29, no.~2, 2009.

\bibitem{wolper2000constructing}
P.~Wolper, ``Constructing automata from temporal logic formulas: A tutorial,''
  in \emph{School organized by the European Educational Forum}.\hskip 1em plus
  0.5em minus 0.4em\relax Springer, 2000, pp. 261--277.

\bibitem{babiak2012ltl}
T.~Babiak, M.~K{\v{r}}et{\'\i}nsk{\`y}, V.~{\v{R}}eh{\'a}k, and
  J.~Strej{\v{c}}ek, ``Ltl to b{\"u}chi automata translation: Fast and more
  deterministic,'' in \emph{International Conference on Tools and Algorithms
  for the Construction and Analysis of Systems}.\hskip 1em plus 0.5em minus
  0.4em\relax Springer, 2012, pp. 95--109.

\bibitem{belta2017formal}
C.~Belta, B.~Yordanov, and E.~A. Gol, \emph{Formal methods for discrete-time
  dynamical systems}.\hskip 1em plus 0.5em minus 0.4em\relax Springer, 2017,
  vol.~89.

\bibitem{clarke1990optimization}
F.~H. Clarke, \emph{Optimization and nonsmooth analysis}.\hskip 1em plus 0.5em
  minus 0.4em\relax SIAM, 1990.

\bibitem{18}
R.~G. Sanfelice, R.~Goebel, and A.~R. Teel, ``Invariance principles for hybrid
  systems with connections to detectability and asymptotic stability,''
  \emph{IEEE Transactions on Automatic Control}, vol.~52, no.~12, pp.
  2282--2297, 2007.

\end{thebibliography}

\end{document}